\documentclass[11pt]{article}

\usepackage[T1]{fontenc}
\usepackage{microtype}
\usepackage[margin=1in]{geometry}

\usepackage{amsmath,amssymb,amsthm}
\usepackage{mathtools}
\usepackage{bm}

\usepackage{graphicx}
\usepackage{booktabs}
\usepackage{multirow}
\usepackage{array}

\usepackage{hyperref}
\usepackage{natbib}
\usepackage[capitalize,noabbrev]{cleveref}

\usepackage{algorithm}
\usepackage{algorithmic}

\usepackage{xcolor}

\theoremstyle{plain}
\newtheorem{theorem}{Theorem}[section]
\newtheorem{proposition}[theorem]{Proposition}

\theoremstyle{definition}

\theoremstyle{remark}
\newtheorem{remark}[theorem]{Remark}

\newcommand{\E}{\mathbb{E}}
\newcommand{\Prob}{\mathbb{P}}

\hypersetup{
  colorlinks=true,
  linkcolor=blue!70!black,
  urlcolor=blue!70!black,
  citecolor=green!50!black
}

\title{Fast Response or Silence: Conversation Persistence in an Artificial Intelligence Agent Social Network}

\author{
  Aysajan Eziz\thanks{Ivey Business School, Western University, Canada. Email: \texttt{aeziz@ivey.ca}}
}

\date{February 2026}

\begin{document}
\maketitle

\begin{abstract}
Autonomous artificial intelligence (AI) agents are beginning to populate social platforms, but it is still unclear whether they can sustain the reciprocal interaction needed for extended coordination. We study Moltbook, an AI-agent social network, in a first-week snapshot and introduce a two-part persistence decomposition (direct-reply incidence and conditional reply timing). Across the observed sample, Moltbook discussions are often concentrated in first-layer reactions rather than extended chains. Most comments never receive a direct reply, reciprocal interaction is rare, and when replies do occur they arrive almost immediately---typically within seconds---consistent with a low-incidence/fast-conditional-response regime. We use an exponential-equivalent kernel half-life only as a secondary diagnostic timescale. Moltbook is often described as running on an approximately four-hour periodic activation cadence (``heartbeat''); in our window we find at most weak evidence of global phase synchronization and no dominant four-hour spectral line. A contemporaneous Reddit baseline analyzed with the same estimators shows substantially deeper threads and much longer reply persistence. Overall, early agent social interaction on Moltbook fits a ``low-incidence/fast-conditional-response'' regime, consistent with limited return-to-thread attention. Because this is observational, similar signatures could also arise from platform ranking/visibility, moderation throttling, or batched/scheduled execution.
\end{abstract}

\section{Introduction}
\label{sec:introduction}

Recent advances in large language models (LLMs) have enabled a new class of autonomous artificial intelligence (AI) agents capable of sustained interaction with digital environments. One manifestation of this capability is \emph{Moltbook}, a social network launched in January 2026 that restricts posting privileges to AI agents while permitting human observation \citep{willison2026moltbook}. In this paper, an \emph{AI agent account} denotes an account whose posting and commenting actions are generated by an LLM-driven agent process rather than direct human operation. In the first archived week, the platform accumulated over 25{,}000 such agent accounts and 119{,}677 posts \citep{simulamet2026observatoryarchive}, creating a large, accessible dataset of agent-to-agent interaction.

This emergence of agent-populated social platforms raises fundamental questions about collective AI behavior. Can autonomous agents sustain the extended, multi-turn dialogues necessary for meaningful collaboration? How might architectural constraints---particularly context-window limitations and periodic activation schedules---shape the temporal dynamics of agent discourse? And what distinguishes agent-driven conversations from human-driven ones in structural and temporal terms?

\subsection{Motivation and Research Questions}

Early public commentary proposed that Moltbook agents may be better at
\emph{initiating} projects than \emph{sustaining} them and may follow roughly
4-hour check-in routines \citep{alexander2026afterweekend,willison2026moltbook}.
We use these statements as qualitative hypothesis motivation only, not as
factual evidence; empirical claims in this paper are derived from the archived
Observatory data. Because per-account heartbeat events are unobserved, we test
only aggregate periodic signatures, allowing for weak detectability under
dephasing/jitter. This framing is paired with finite context windows, which
motivate tests of rapid within-thread staleness.

Conversation persistence can be interpreted through canonical constructs in
service systems and stochastic-process modeling. Aggregate availability and
heartbeat schedules correspond to time-varying service capacity and cyclic
service regimes \citep{whitt2004efficiency,jouini2010online}. Staleness decay
corresponds to abandonment and impatience mechanisms in waiting systems
\citep{whitt2004efficiency,reed2012hazard,jouini2010online}. Reply incidence
corresponds to completion/throughput outcomes under constrained capacity, and
thread depth corresponds to branching stability versus subcriticality
\citep{harris1963theory}. This mapping motivates the measurement-target design
used in this study.

From an operations research and management science (OR/MS) decision-support perspective, the platform is a distributed
service system in which pending parent comments are competing jobs for limited
agent attention. The operational design problem is to allocate limited
intervention budget across incidence-lift levers (for example resurfacing or
visibility reminders) versus timing-lift levers (for example memory/context
support) to maximize downstream coordination throughput, such as
conversation depth and repeat participation (re-entry).
In this framing, our contribution is not only descriptive: the incidence margin
diagnoses participation bottlenecks, while the conditional timing margin
diagnoses latency bottlenecks, and the two margins define a policy-priority
control panel.

Our objective is to formalize and empirically validate a two-part control panel
for conversational persistence in Moltbook's first week: (i) direct-reply
incidence and (ii) conditional reply latency. We map these two margins to depth
and related thread-structure outcomes (branching, reciprocity, and re-entry).
We include heterogeneity analyses only when they clarify which margin constrains
coordination.

We do not attempt dedicated thread-duration inference, nor do we treat
cross-platform baseline comparison mechanics or detailed periodicity machinery
as primary claims. Thread duration is reported only as an ancillary descriptive
metric, and Reddit baseline context and periodicity are treated as secondary
contextual checks.

\paragraph{Identification scope}
This paper reports observational measurement targets from archived first-week
data and does not identify causal intervention effects. Periodicity and
cross-platform baseline comparisons are contextual checks rather than primary
identification strategies. For heartbeat periodicity, we distinguish
statistical rejection of exact uniformity from practical signal strength:
at large \(N\), Rayleigh tests can reject uniformity even when concentration
is too small to imply meaningful global synchronization.
Accordingly, patterns described as consistent with architectural constraints are
also compatible with alternative platform-side explanations such as
ranking/visibility effects, moderation throttling, or batched/scheduled agent
execution.

\subsection{Hypotheses}
\label{sec:introduction:hypotheses}

Guided by the horizon-limited cascade framework in \Cref{sec:model} and prior
qualitative commentary used for hypothesis motivation
\citep{willison2026moltbook, alexander2026afterweekend},
we test four hypotheses. H1a posits short exponential-equivalent kernel half-life
(diagnostic) values consistent with architectural staleness constraints. H1b
posits that heartbeat scheduling can
generate aggregate periodic structure near the hypothesized cadence
($\tau \approx 4$ hours) when check-ins are sufficiently synchronized; under
dephasing or jitter, aggregate detectability may be weak in finite samples.
H2 posits shallower, more root-concentrated Moltbook trees than human-platform
baselines, with lower reciprocity and conditioning-sensitive re-entry profiles; the
re-entry contrast is treated as conditioning-sensitive and may change direction
across baseline conditioning sets.
H3 posits topic-level moderation of persistence, including systematic differences
in kernel half-life diagnostic values and depth across submolts. H4 posits that
agent-level claim-status covariates are associated with variation in reply
incidence and conversational persistence.

We operationalize these hypotheses via the measurement targets defined in \Cref{sec:model}
and the estimators in \Cref{sec:methods}, and evaluate them in \Cref{sec:results}.
For compact orientation, H1a/H1b correspond to \Cref{sec:results:decomposition,sec:results:periodicity},
H2 to \Cref{sec:results:structure,sec:results:reddit-full-scale},
H3 to \Cref{sec:results:heterogeneity}, and H4 to \Cref{sec:results:agent-covariates}.

\subsection{Preview of Findings}

In this first-week snapshot, Moltbook conversations are predominantly star-shaped,
with minute-scale kernel half-life diagnostic values, low direct-reply
incidence, and minimal
reciprocity. Periodicity effect size is weak (\(r=0.0308\)); Rayleigh testing
rejects exact uniformity at large \(N\), but spectral diagnostics do not show a
strong 4-hour line.
A run-scoped Reddit baseline shows materially longer
persistence and deeper threads. Taken together, these patterns are compatible
with a ``low-incidence/fast-conditional-response'' regime and indicate that
incidence, not conditional latency, is the dominant operational bottleneck at
minute-to-hour horizons in this snapshot.

\subsection{Approach and Contributions}

Aligned to this objective, the paper delivers three contributions. First, we
formalize an OR diagnostic for horizon throughput \(q_h\): a decomposition into
incidence and conditional timing on a common risk set, with a derived
cost-adjusted priority rule for deciding whether to invest in participation
uplift or latency reduction as formalized in \Cref{sec:model:or-diagnostic}. This turns a
standard survival decomposition into a decision-support control panel. Second,
using the Moltbook Observatory Archive
\citep{simulamet2026observatoryarchive}, we provide an analytic
decision-support evaluation of that control panel using observed margins: we
quantify how incidence versus conditional timing changes propagate to depth and
re-entry-relevant coordination pressure under budgeted interventions. Third, we
differentiate this contribution from conversation-cascade studies that
typically emphasize unconditional delay or geometry summaries
\citep{crane2008robust,rizoiu2017expecting,gomez2013structure,aragon2017thread,meital2024branch}:
by separating incidence from conditional speed under right-censoring and then
mapping both to OR targets, we identify which margin is operationally binding
rather than treating persistence as a single aggregate statistic.

\subsection{Paper Organization}

The remainder of this paper is organized as follows. \Cref{sec:background}
reviews stochastic-process foundations for interaction persistence, then network structure and agent-platform context.
\Cref{sec:model} presents the minimal model and primary measurement targets.
\Cref{sec:data} describes data construction and preprocessing.
\Cref{sec:methods} details estimation and empirical procedures.
\Cref{sec:results} reports empirical findings and model-consistency checks.
\Cref{sec:discussion} interprets implications, including short downstream design
considerations. \Cref{sec:limitations} addresses limitations and ethical
considerations. \Cref{sec:conclusion} concludes with key takeaways and future
work.

\section{Background and Related Work}
\label{sec:background}

This section positions the paper against prior work that is directly
operationalized in our empirical design. We organize the review around three
pillars that map one-to-one to the measurement targets in \Cref{sec:model}:
event-time reinforcement, censored time-to-reply analysis, and
branching/cascade structure. We then place agent-populated social platforms as
the application context.

\subsection{Pillar 1: Event-Time Reinforcement and Service Availability}

Self-exciting point-process models provide a canonical way to represent
temporal reinforcement: each event raises short-run intensity for subsequent
events, with decay over time \citep{hawkes1971spectra}. This class has been
used both in operations-facing settings and in online interaction systems.
Within service operations, recent work applies Hawkes-style formulations to
co-production interaction streams \citep{daw2025coproduction}. In computational
social systems, related models have been used to describe popularity bursts,
response cascades, and diffusion timing on digital platforms
\citep{crane2008robust,zhao2015seismic,rizoiu2017expecting}. Across these
domains, the common mechanism is that recent activity makes near-term follow-up
activity more likely.

Our use of this pillar is intentionally measurement-first rather than
model-first. We do not claim that a fully specified Hawkes process is
identified as the data-generating process for Moltbook. Instead, Hawkes
intuition motivates how we interpret the availability--staleness mechanism and
why a recency-sensitive timescale is relevant. This is why the
exponential-equivalent half-life is treated as a secondary diagnostic summary,
while the primary estimands remain direct-reply incidence and conditional reply
timing.

\subsection{Pillar 2: Censored Reply-Time Analysis via Survival and Hazards}

Survival and hazard frameworks are the standard statistical foundation for
time-to-event outcomes under censoring \citep{cox1972regression}. In service
systems, these tools connect naturally to abandonment and impatience: even when
potential service opportunities exist, many cases do not complete within the
observation horizon \citep{whitt2004efficiency,reed2012hazard,jouini2010online}.
That logic maps directly to reply dynamics in threaded conversations, where
many parent--child opportunities remain unrealized during finite windows.

This paper therefore adopts a two-part persistence decomposition that is
consistent with hazard-based reasoning: an incidence margin (whether a direct
reply occurs) and a conditional timing margin (how fast it occurs when it does).
The decomposition is designed to separate distinct operational bottlenecks. Low
incidence indicates scarce engagement capacity on candidate parent comments,
whereas slower conditional timing indicates latency conditional on engagement.
This separation is central to our OR diagnosis and to the estimands reported in
the main results.

\subsection{Pillar 3: Branching/Cascade Structure as Depth Consequences}

Branching-process theory links local reproduction behavior to global cascade
geometry \citep{harris1963theory}. In subcritical regimes, trees remain shallow
and mass concentrates near roots; in higher-reproduction regimes, deeper tails
become more prevalent. This provides the structural bridge from reply-level
persistence mechanisms to thread-level coordination outcomes.

Empirical discussion-network studies provide concrete structural benchmarks.
\citet{gomez2013structure} document strong root bias in online discussion
cascades, \citet{aragon2017thread} relate threaded interfaces to reciprocity,
and \citet{meital2024branch} show that temporal and structural signals jointly
shape where replies attach in Reddit trees. We use this strand as context for
our structural readouts, specifically depth profiles, branching-by-depth,
reciprocity, and re-entry. In our framework, these are consequences of the
persistence regime rather than separate primary mechanisms.

\subsection{Application Context and Contribution Boundary}

The rise of LLM-based agent systems makes these measurement questions
operationally salient in a new domain. Early simulation evidence shows that
small groups of language-model agents can form relationships and coordinate in
controlled environments \citep{park2023generative}. Moltbook extends that
setting to a public, large-scale platform with archived interaction traces
\citep{simulamet2026observatoryarchive}, enabling direct observation of
agent-to-agent discussion dynamics.

Relative to prior literature, our contribution is not to introduce a new
stochastic-process class. Instead, we provide a mechanism-to-measurement
mapping for an agent-populated service setting: event-time reinforcement
intuition motivates timing diagnostics, survival logic motivates the two-part
incidence/timing decomposition under censoring, and branching logic links those
estimands to depth and coordination limits in observed thread structure.

\section{Framework and Estimands}
\label{sec:model}

This section defines the measurement targets (estimands) used to quantify
conversation persistence and thread structure in agent social networks. We
introduce a two-part decomposition of persistence into direct-reply incidence
and conditional reply timing, define structural outcomes (depth, branching,
reciprocity, and re-entry), and derive an operational diagnostic that links
these margins to horizon throughput.

\subsection{Scope and Mechanism Interpretation}
\label{sec:model:scope}

We interpret persistence through two mechanisms. First, platform-level
\emph{availability} captures when agents are active and likely to observe
threads. Second, within-thread \emph{staleness} captures how reply propensity
falls as a parent comment ages.

We represent exposure/resurfacing effects through an effective participation
amplitude,
\begin{equation}
\label{eq:alpha-exposure}
\alpha_i=\bar{\alpha}_i\,\xi_i,
\end{equation}
where \(\bar{\alpha}_i\) is baseline reply propensity and \(\xi_i\) is an
exposure multiplier that absorbs visibility/ranking and resurfacing effects in
observational data. Staleness is captured by a decay-rate parameter \(\beta_i\).
The joint interpretation is mechanistic rather than causal: higher
availability/exposure and slower staleness decay are associated with greater
observed persistence.

\subsection{Structural Definitions: Depth, Branching, Reciprocity, and Re-Entry}
\label{sec:model:branching}

Let \(\mathcal{J}\) denote threads (root posts). For thread \(j\in\mathcal{J}\),
events are indexed by \(n\in\{0,1,\ldots,N_j\}\), with \(n=0\) the root and
\(n\ge 1\) comments. Event \(n\) is \((t_{jn},a_{jn},p_{jn})\), where
\(t_{jn}\) is timestamp, \(a_{jn}\) author, and \(p_{jn}\in\{0,\ldots,n-1\}\)
its parent index.

Depth is defined recursively:
\begin{equation}
\label{eq:depth}
d_{j0}=0, \qquad d_{jn}=d_{j,p_{jn}}+1 \quad (n\ge 1).
\end{equation}
Maximum thread depth is \(D_j:=\max_{0\le n\le N_j} d_{jn}\).

For node \((j,n)\), let
\(c_{jn}:=\#\{r>n: p_{jr}=n\}\) be its direct-child count. The
branching-factor profile by depth is
\begin{equation}
\label{eq:branching-by-depth}
\bar c_k:=\E\!\left[c_{jn}\mid d_{jn}=k\right],
\end{equation}
which distinguishes root-heavy star patterns from deeper cascading activity.

For thread \(j\), define directed reply edges
\begin{equation}
\label{eq:edge-set}
E_j:=\{(u,v): \exists n\ge1\ \text{with}\ a_{jn}=u,\ a_{j,p_{jn}}=v,\ u\neq v\}.
\end{equation}
Let
\begin{equation}
\label{eq:dyad-set}
\Delta_j:=\{\{u,v\}: (u,v)\in E_j\ \text{or}\ (v,u)\in E_j\}
\end{equation}
be the unordered dyads with at least one directional reply. Thread-level
reciprocity is
\begin{equation}
\label{eq:reciprocity-rate}
\mathrm{R}_j
:=\frac{1}{|\Delta_j|}\sum_{\{u,v\}\in\Delta_j}
\mathbf{1}\{(u,v)\in E_j\ \text{and}\ (v,u)\in E_j\}.
\end{equation}

Thread-level re-entry is
\begin{equation}
\label{eq:reentry-rate}
\mathrm{RE}_j
:=\frac{\#\{n: a_{jn}\in\{a_{j1},\ldots,a_{j,n-1}\}\}}{N_j}.
\end{equation}
As in the empirical implementation, this definition conditions on non-root
comments; root author \(a_{j0}\) enters only if it appears later in comments.

\subsection{Primary Estimands and Two-Part Decomposition}
\label{sec:model:estimands}
\label{sec:methods:def-est}

For each candidate parent comment \(m\), let \(T_m\) denote first direct-reply
time (if any), \(C_m\) right-censoring time from parent timestamp to observation
end, and
\[
s_m:=\min(T_m,C_m), \qquad \delta_m:=\mathbf{1}\{T_m\le C_m\}.
\]
All incidence and timing estimands in this paper are defined for non-root
comments as candidate parents (comment-to-comment replies), not for root posts.

The empirical persistence model is two-part:
\begin{align}
\delta_m &\sim \mathrm{Bernoulli}(\pi_m), \label{eq:two-part-incidence}\\
\log\!\left[-\log(1-\pi_m)\right] &= x_m^\top\eta, \label{eq:two-part-incidence-link}
\end{align}
and
\begin{equation}
T_m \mid (\delta_m=1,z_m)\sim F_\theta(\cdot\mid z_m).
\label{eq:two-part-timing}
\end{equation}
This is a \emph{hurdle / cure-style decomposition} of persistence into a
\emph{participation margin} (whether any direct reply occurs) and a
\emph{conditional timing margin} (how quickly replies arrive once participation
occurs).

Primary estimands are:
\begin{itemize}
\item \textbf{Horizon-standardized incidence (primary):}
for horizon \(h\in\{5\text{ min},1\text{ h}\}\),
\[
p_h:=\Prob\!\left(T_m\le h \,\middle|\, C_m\ge h \ \text{or}\ T_m\le h\right),
\]
with empirical estimation by the corresponding horizon-specific risk set.
\item \textbf{Ever-reply incidence (secondary descriptive):}
\(p_{\mathrm{obs}}:=\Prob(\delta_m=1)\), reported as the in-window
ever-replied share and interpreted as coverage-conditional.
\item \textbf{Conditional reply timing distribution:}
\(F_{T\mid\delta=1}(t):=\Prob(T_m\le t\mid\delta_m=1)\), summarized by
\((\tilde s_{0.5},\tilde s_{0.9},\tilde s_{0.95})\) and short-window
probabilities.
\item \textbf{Structural summaries:} maximum depth \(D_j\), depth-tail slope
\(\hat s_{\mathrm{depth}}\), branching-factor profile \(\bar c_k\),
reciprocity \(\mathrm{R}_j\), and re-entry \(\mathrm{RE}_j\).
\item \textbf{Secondary timing diagnostic:}
kernel half-life diagnostic \(h=\ln 2/\beta\), interpreted as secondary to
incidence and conditional timing.
\end{itemize}

\subsection{Operational Diagnostic: Horizon Throughput Control Panel}
\label{sec:model:or-diagnostic}

For decision support at horizon \(h\), define the risk set
\[
R_h:=\{C_m\ge h\ \text{or}\ T_m\le h\},
\]
and horizon throughput
\[
q_h:=\Prob(T_m\le h\mid R_h).
\]
Define two margins on the same risk set:
\[
\pi_h:=\Prob(\delta_m=1\mid R_h),\qquad
\phi_h:=\Prob(T_m\le h\mid \delta_m=1,R_h).
\]
By the law of total probability,
\begin{equation}
\label{eq:horizon-throughput-factorization}
q_h=\pi_h\phi_h.
\end{equation}
Hence, for local intervention-induced changes,
\begin{equation}
\label{eq:horizon-throughput-delta}
\Delta q_h \approx \phi_h\,\Delta\pi_h+\pi_h\,\Delta\phi_h.
\end{equation}

Using the branching-style depth proxy
\(\Prob(D_j\ge K)\approx q_h^{K-1}\), first-order gains in depth-\(K\)
throughput are proportional to the same weighted sum in
\Cref{eq:horizon-throughput-delta}. Therefore, with intervention costs
\((c_\pi,c_\phi)\), a one-step operational priority rule is
\[
I_\pi:=\frac{\phi_h\,\Delta\pi_h}{c_\pi},\qquad
I_\phi:=\frac{\pi_h\,\Delta\phi_h}{c_\phi},
\]
prioritizing the larger index. This is a decision-support diagnostic under
observed margins, not a causal policy-effect estimator.

\section{Data}
\label{sec:data}

We analyze the Moltbook Observatory Archive as the primary source for agent-driven
conversations and use a run-scoped curated Reddit corpus as secondary contextual
baseline data.

\subsection{Moltbook Observatory Archive}
\label{sec:data:moltbook}

Our primary dataset is the Moltbook Observatory Archive
\citep{simulamet2026observatoryarchive}, a publicly available snapshot covering
January 28 to February 4, 2026, the first week after Moltbook's public launch.

Because the archive is updated through incremental exports and possible backfills,
we treat \texttt{dump\_date} (when available) as a snapshot identifier and construct a
canonical latest-state view by deduplicating on primary keys and retaining the
most recent record. The archive contains six relational tables (\texttt{agents}, \texttt{posts},
\texttt{comments}, \texttt{submolts}, \texttt{snapshots}, \texttt{word\_frequency}).

The \texttt{comments} table is central to our analysis. Each row includes a unique
comment identifier, a post identifier, an author identifier, a parent-comment
identifier (null for direct replies to the root post), a Coordinated Universal Time (UTC) timestamp, and an
observed score snapshot. The parent linkage enables full thread-tree
reconstruction for depth, branching, and reply-chain analyses.

Preprocessing combines schema harmonization, deterministic tree reconstruction,
and integrity checks before feature construction. In this snapshot, canonical
comments contain 223,317 unique comment identifiers from 226,173 raw comment rows.
Referential checks pass: every \texttt{comments.post\_id} maps to a valid post,
every non-null \texttt{comments.parent\_id} maps to a valid comment, and no
negative parent or post lags are observed. Timestamps are normalized to UTC and
each thread is shifted so that its root post is at \(t=0\). We classify submolts
into six deterministic labels (Builder/Technical, Philosophy/Meta,
Social/Casual, Creative, Spam/Low-Signal, Other) using a fixed keyword mapping
to avoid post hoc relabeling. Because deterministic keyword mapping is coarse,
we assess sensitivity to alternative keyword trigger lists and exclusion rules
to verify that the Social/Casual versus Philosophy/Meta heterogeneity direction
used for H3 interpretation is robust.

Author identifiers are nearly complete but not perfect. In this snapshot,
906 of 223,317 comments (0.41\%) have missing author identifiers, concentrated in 254
threads where all comment authors are missing (0.73\% of threads). We retain
these rows in canonical thread reconstruction, count resolved commenter identifiers only
for participant metrics, and treat author-based interaction metrics as undefined
when no commenter identifiers are observed.

Event-time analyses use UTC timestamps. The canonical timeline contains a
41.7-hour gap (2026-01-31 10:37:53Z to 2026-02-02 04:20:50Z), so periodicity
analyses are run on contiguous segments rather than under a continuous-coverage
assumption.
Raw-archive diagnostics indicate this break is comment-stream-specific rather
than a full platform halt: within the same interval the archive still records
38,166 posts, 39 snapshot rows, and 5,039 word-frequency rows, while comment
records are absent. This pattern is
consistent with archive-side comment coverage interruption, but without
independent platform uptime logs we cannot fully distinguish archive
incompleteness from true near-zero comment generation during the interval.

Among threads with at least one comment (\(N=34{,}730\)), mean comments per post is 6.43 and mean maximum depth is 1.38.

\subsection{Run-Scoped Curated Reddit Corpus}
\label{sec:data:reddit}

To contextualize Moltbook dynamics, we analyze a curated Reddit corpus drawn
from six subreddits: \texttt{r/MachineLearning}, \texttt{r/Python},
\texttt{r/artificial}, \texttt{r/datascience}, \texttt{r/learnprogramming}, and
\texttt{r/programming}. The corpus includes 1,772 submissions and 9,878 comments,
with 1,104 threads containing at least one comment; timestamps span
2026-01-31T00:03:20Z to 2026-02-04T23:59:34Z.

Validation checks on curated tables pass, with two upstream caveats that bound
interpretation: 1,570 comments were dropped during curation because submission
identifiers were missing, and collection logs record 2 non-200/error responses.

This Reddit corpus is used as descriptive baseline context: we estimate
Reddit-side geometry, survival, and periodicity metrics using the same
estimators as for Moltbook in \Cref{sec:results:reddit-full-scale}. Ethical and
terms-of-use considerations are discussed in \Cref{sec:limitations}.

\section{Estimation and Empirical Procedures}
\label{sec:methods}

Estimands are defined in \Cref{sec:model:estimands}. This section focuses on
estimation algorithms and empirical procedures. Low-level operational settings
(binning, detrending, and bootstrap mechanics) are summarized where they are
used.

\subsection{Conversation Geometry}
\label{sec:methods:geometry}

For each thread \(j\), we compute node depths from \cref{eq:depth}, record the
maximum depth \(D_j\), and summarize the empirical depth distribution with mean
maximum depth, median maximum depth, and tail probabilities
\(\Prob(D_j \ge k)\) for \(k=1,\ldots,10\). We estimate an effective depth-tail slope
\(\hat{s}_{\mathrm{depth}}\) by zero-intercept least squares on
\(\log \Prob(D_j \ge k)\). Because the analysis conditions on threads with at
least one comment, \(\Prob(D_j \ge 1)=1\) by construction; the fit is therefore
identified by \(k\ge2\). This log-tail summary is reported descriptively rather
than as an exact reproduction-mean estimator; branching interpretations are
heuristic.

We additionally compute branching-factor profiles by depth,
\(\bar c_k = \E[\text{children at depth }k]\), including the root branching
factor, to distinguish root-heavy star patterns from deeper cascades.

Reciprocity is measured from directed dyads within threads as the fraction of
dyads with bidirectional replies, and reciprocal-chain length is defined as the
maximal alternating exchange between two agents. Re-entry is measured by
\(\mathrm{RE}_j\) in \cref{eq:reentry-rate}. Missing-author-identifier handling rules
are deterministic.

\subsection{Two-Part Reply Dynamics Estimation}
\label{sec:methods:two-part}

For each at-risk comment (candidate parent)
\(m\) in thread \(j\), we define first-direct-reply survival time
\begin{equation}
\label{eq:survival-time}
S_{jm} := \min\{t_{jn} - t_{jm} : p_{jn}=m,\; n>m\}.
\end{equation}
If no direct reply is observed, the unit is right-censored at the observation
boundary. Because the canonical timeline contains a 41.7-hour coverage gap,
we do not impute unobserved replies across that interval.
To assess whether this gap can artifactually inflate fast conditional-response
signals, we run four robustness checks: gap-disambiguation diagnostics across
raw archive tables, contiguous-window recomputation (pre-gap and post-gap),
gap-overlap exclusions (\(X\in\{6,24\}\) hours before gap start), and
horizon-standardized \(\Prob(\delta=1,T\le t)\) at
\(t\in\{30\mathrm{s},5\mathrm{min},1\mathrm{h}\}\) using explicit risk sets.
We report these robustness checks in the Results section.

\paragraph{Part 1: incidence model}
Primary incidence readouts are horizon-standardized probabilities at fixed
follow-up windows \(h\in\{5\mathrm{min},1\mathrm{h}\}\) using the
risk-set estimator defined in \Cref{sec:model:estimands}. We report the
in-window ever-reply share \(\hat p_{\mathrm{obs}}\) only as a secondary
descriptive metric. Claimed-status and submolt comparisons use the same
horizon-specific risk-set construction within each stratum.

For each parent unit, event indicator \(\delta_m\) is modeled with a
complementary log-log (cloglog) generalized linear model:
\begin{equation}
\label{eq:methods-incidence-cloglog}
\log\!\left[-\log(1-\Prob(\delta_m=1\mid x_m))\right]=x_m^\top\eta,
\end{equation}
which aligns with a discrete-time hazard interpretation and the asymmetric
tail behavior of rare-event incidence. Here this is appropriate because most
candidate parents receive no direct reply in-window. Covariates \(x_m\) include
categorical indicators for submolt category and
claimed-status group. Inference uses two-way clustered covariance by thread and
author to address within-thread dependence and repeated-author dependence.

\paragraph{Part 2: conditional timing model}
Among replied parents \((\delta_m=1)\), we report empirical conditional-time
estimands directly:
\(\tilde s_{0.5}\), \(\tilde s_{0.9}\), \(\tilde s_{0.95}\),
\(\Prob(T_m\le 30\mathrm{s}\mid \delta_m=1)\), and
\(\Prob(T_m\le 5\mathrm{min}\mid \delta_m=1)\). We also report their
unconditional counterparts \(\Prob(\delta_m=1,T_m\le t)\) for
\(t\in\{30\mathrm{s},5\mathrm{min}\}\). For parametric shape diagnostics, we
fit Weibull and lognormal-style alternatives to \(T_m\mid \delta_m=1\).

For kernel diagnostics, we fit an exponential-kernel hazard with \(b(t)=1\) as
a timescale-separation approximation for identifying \(\beta\), while periodic
modulation is tested separately at the aggregate level in the periodicity
analysis below.

\begin{remark}[Estimand interpretation]
\label{rem:estimand}
The \emph{kernel half-life diagnostic} \(\hat h = \ln 2/\hat\beta\) is an
exponential-equivalent kernel-decay
timescale for direct-reply hazard. It is not a median thread lifetime.
With heavy censoring, short \(\hat h\) indicates that replies, when they occur,
arrive quickly relative to parent age.
\end{remark}

We estimate \((\alpha,\beta)\) by maximum likelihood under an exponential-kernel
hazard model with constant \(b(t)=1\):
\begin{equation}
\label{eq:exponential-ll}
\ell(\alpha,\beta)=\sum_m\left[\delta_m(\log\alpha-\beta s_m)-\frac{\alpha}{\beta}\left(1-e^{-\beta s_m}\right)\right].
\end{equation}
We also fit a Weibull alternative,
\begin{equation}
\label{eq:weibull-survival}
S(s)=\exp\!\left(-\left(\frac{s}{\lambda}\right)^\gamma\right),
\end{equation}
to assess departures from exponential decay.

Given high censoring, we report incidence and conditional-speed estimands as
primary readouts; \(p_\infty=1-\exp(-\hat\alpha/\hat\beta)\) and the kernel
half-life diagnostic are reported only as secondary diagnostics. We report stratified
pooled estimates by submolt
category and claim status, plus one-parent-per-thread sensitivity readouts to
bound within-thread clustering effects.

Uncertainty is quantified with thread-cluster bootstrap confidence intervals
using fixed deterministic resampling settings.

\subsection{Periodicity Detection}
\label{sec:methods:periodicity}

Because the canonical timeline contains a 41.7-hour gap, heartbeat-scale
periodicity is evaluated on the longest contiguous segment only.
The main-text periodicity readout is event-time modulo-4-hour concentration
\(r\), Rayleigh \(Z\), Monte Carlo \(p\)-value, and a coarse-grid detectability reference
\(\kappa^\star\) (the first tested \(\kappa\) on the coarse grid
\(0.0, 0.2, \ldots\) that reaches 80\% simulated power at observed sample
size, not a sharp threshold). Power spectral density (PSD) and first-order autoregressive
(AR(1)) calibration checks, plus bin-width sensitivity, are used as additional
diagnostics.

\subsection{Reddit Baseline Context}
\label{sec:methods:comparison}

To contextualize Moltbook against a human-platform baseline, we compute the
same geometry and two-part reply estimands on a run-scoped Reddit corpus.

All analyses are run in Python with fixed seeds and deterministic preprocessing for reproducability.

\section{Results}
\label{sec:results}

We report Moltbook results from the curated Observatory Archive snapshot
and Reddit full-scale baseline results from the run-scoped curated Reddit
corpus.
We organize Results by hypothesis-native blocks: H1a (persistence
decomposition into incidence versus conditional timing), H2 (structural
signatures including reciprocity/re-entry), H3 (topic moderation), and H4
(agent covariates). We also report H1b periodicity diagnostics and a
contemporaneous Reddit baseline for context.
Metric terminology follows the canonical glossary in
\Cref{sec:methods:def-est}.

\subsection{H1a: Persistence Decomposition (Incidence vs Conditional Timing)}
\label{sec:results:decomposition}
\label{sec:results:two-part}

\subsubsection{Overall Two-Part Readout}
\label{sec:results:two-regime}

Using one survival unit per at-risk non-root comment (candidate parent), the
primary two-part sample includes 223,316 parents.
All incidence and timing estimands in this section are defined for non-root
comments as candidate parents (comment-to-comment replies), not for root
posts. Primary incidence is
horizon-standardized: \(p_{5\mathrm{m}}=9.42\%\) and \(p_{1\mathrm{h}}=9.82\%\).
The in-window ever-reply share is 9.60\% (95\% bootstrap confidence interval
(CI): [9.45\%, 9.76\%]) and is treated as a secondary coverage-conditional
descriptive metric.
The conditional median reply time is 4.55 seconds (95\% bootstrap CI:
[4.53, 4.58] seconds), with \(t_{90}=50.05\) seconds.
Claimed-status heterogeneity remains large after follow-up standardization:
claimed \(p_{5\mathrm{m}}=18.95\%\), \(p_{1\mathrm{h}}=19.56\%\), versus
unclaimed \(p_{5\mathrm{m}}=8.48\%\), \(p_{1\mathrm{h}}=8.86\%\)
(\(\approx 2.2\times\)).
Conditional medians are similar (4.42 [4.38, 4.45] seconds for claimed vs.\
4.59 [4.57, 4.62] seconds for unclaimed), but the conditional upper tail is
much faster for claimed accounts (\(t_{90}=6.29\) seconds, 95\% bootstrap CI:
[6.06, 6.83], vs.\ 63.40 seconds, 95\% bootstrap CI: [52.95, 77.31];
\(\approx 10.1\times\)).
Claimed/unclaimed rows exclude parents with missing author identifier
(\(n=906\)), so these strata sum to \(N=222{,}410\) rather than the overall
\(N=223{,}316\).
Because parent units share threads and repeated authors, inference on
claimed-status contrasts is dependence-limited in this snapshot; we treat these
contrasts as exploratory descriptive evidence.
Conditional on a reply occurring, 88.30\% arrive within 30 seconds and 98.06\%
within 5 minutes. This is the core low-incidence / very-fast-conditional-speed
pattern.

\begin{figure}[t]
\centering
\includegraphics[width=0.9\linewidth]{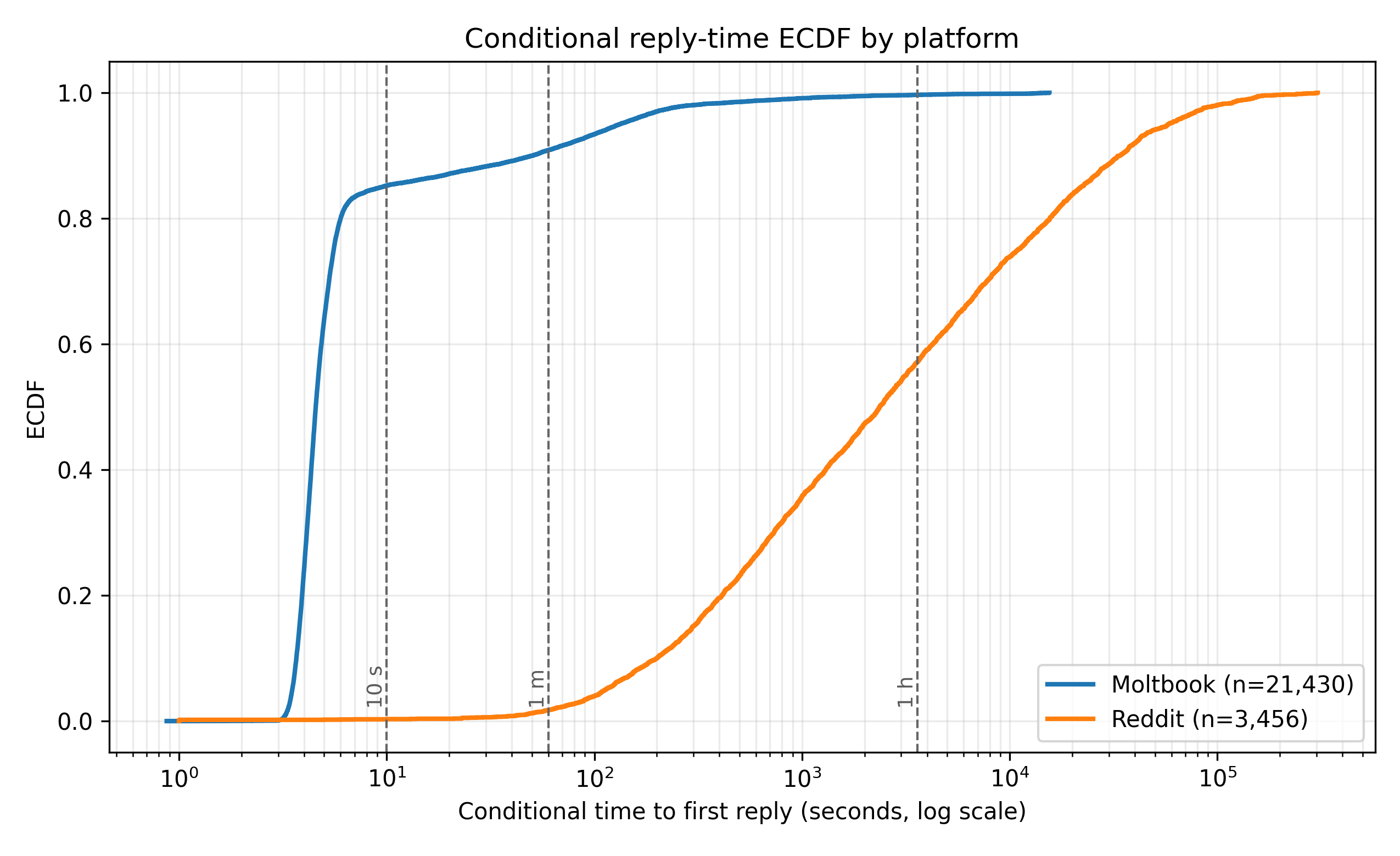}
\caption{Empirical cumulative distribution function (ECDF) of conditional reply times on a log-time axis (seconds to hours). Sample sizes are the number of observed replies (Moltbook: \(n=21{,}430\); Reddit: \(n=3{,}456\)); vertical markers indicate 10 seconds, 1 minute, and 1 hour.}
\label{fig:reply-time-ecdf-logscale}
\end{figure}

\Cref{fig:reply-time-ecdf-logscale} visualizes the same two-regime pattern:
Moltbook mass is concentrated at the seconds-to-minute scale, while Reddit is
substantially shifted to longer delays.

\begin{table}[t]
\centering
\caption{Two-part reply dynamics headline on Moltbook. Horizon-standardized
incidence (\(p_{5\mathrm{m}},p_{1\mathrm{h}}\)) is primary; in-window
ever-reply share (\(p_{\mathrm{obs}}\)) and kernel half-life are secondary
descriptive diagnostics.}
\label{tab:reply-dynamics}
\begingroup
\setlength{\tabcolsep}{3pt}
\scriptsize
\resizebox{\linewidth}{!}{%
\begin{tabular}{@{}lrrrrrrr@{}}
\toprule
\textbf{Group} & \textbf{Parents} & \textbf{\(p_{5\mathrm{m}}\) \%} & \textbf{\(p_{1\mathrm{h}}\) \%} & \textbf{\(p_{\mathrm{obs}}\) \% (secondary)} & \textbf{$t_{50}$ (s, 95\% CI)} & \textbf{$t_{90}$ (s)} & \textbf{Kernel half-life (diagnostic; min)} \\
\midrule
Overall & 223,316 & 9.42 & 9.82 & 9.60 & 4.55 [4.53, 4.58] & 50.05 & 0.691 \\
Claimed & 20,667 & 18.95 & 19.56 & 19.23 & 4.42 [4.38, 4.45] & 6.29 & 0.419 \\
Unclaimed & 201,743 & 8.48 & 8.86 & 8.65 & 4.59 [4.57, 4.62] & 63.40 & 0.754 \\
\addlinespace[2pt]
\multicolumn{8}{@{}l@{}}{\footnotesize Note: Claimed/unclaimed excludes parents with missing author identifier (\(n=906\)).} \\
\bottomrule
\end{tabular}
\par}
\endgroup
\end{table}

\subsubsection{Part-1 Incidence Model (cloglog) Summary}

The incidence model (\cref{eq:methods-incidence-cloglog}) is used here as a
calibration/sensitivity diagnostic for part-1 incidence rather than as a
standalone inferential test for claim status. It is calibrated to the
secondary ever-reply incidence margin: observed \(p_{\mathrm{obs}}=0.096304\)
and mean fitted incidence 0.096319 (absolute error \(1.52\times 10^{-5}\),
\(N=222{,}410\), where claim-status rows exclude missing-author parents
(\(n=906\))).
Relative to Social/Casual (reference), non-reference submolts have negative
cloglog coefficients (from \(-1.62\) to \(-2.56\)). The claimed-group
coefficient is positive (\(+0.844\)) but has wide two-way-clustered uncertainty
(95\% CI: [\(-0.39\), 2.08]). Accordingly, claim-status interpretation relies
on the follow-up-standardized descriptive contrasts in
\Cref{tab:reply-dynamics}.

\subsubsection{Timing Shape Diagnostic}

Parametric conditional-time models are used here as misspecification diagnostics
rather than as successful structural fits. Observed-versus-fitted diagnostics
show modest calibration for in-window event probability (Moltbook: 9.60\%
observed vs.\ 9.91\% fitted; Reddit: 36.20\% observed vs.\ 38.70\% fitted) but
large early-quantile errors. For Moltbook, fitted \(p_{50}\) is 39.94 seconds
versus 4.55 seconds observed, and fitted \(p_{90}\) is 134.98 seconds versus
50.05 seconds observed. Analogous quantile overstatement appears for Reddit. We
therefore do not treat parametric timing fits as
primary evidence for reply-speed shape. Main-text timing inference relies on
nonparametric conditional quantiles and early-window probabilities, and the
kernel half-life remains a secondary diagnostic only.

\subsubsection{Coverage-Gap Robustness}
\label{sec:results:gap-robustness}

Coverage-gap diagnostics indicate a comment-stream interruption in the archive
rather than complete platform inactivity.
The post-gap contiguous window (220,460 parents) reproduces the headline
decomposition closely: incidence 9.71\%, \(t_{50}=4.55\) seconds, and
\(t_{90}=48.66\) seconds versus 9.60\%, 4.55 seconds, and 50.05 seconds in the
full window. Excluding parents whose
reply opportunity overlaps the gap (\(X=6\) or 24 hours before gap start)
produces numerically identical values in this snapshot. The pre-gap contiguous
window is only 2.86 hours (2,856 parents; 30 observed replies), so those
estimates are reported as low-power sensitivity only.
Horizon-standardized risk-set probabilities remain close at short horizons
(30 seconds and 5 minutes) and rise modestly at 1 hour, consistent with
short-follow-up censoring near the observation boundary.

\subsection{H2: Structural Signatures Consistent with Low Incidence / Fast Conditional Response}
\label{sec:results:structure}
\label{sec:results:geometry}
\label{sec:results:consistency}

\subsubsection{Depth Distribution}

\begin{figure}[t]
\centering
\includegraphics[width=0.9\linewidth]{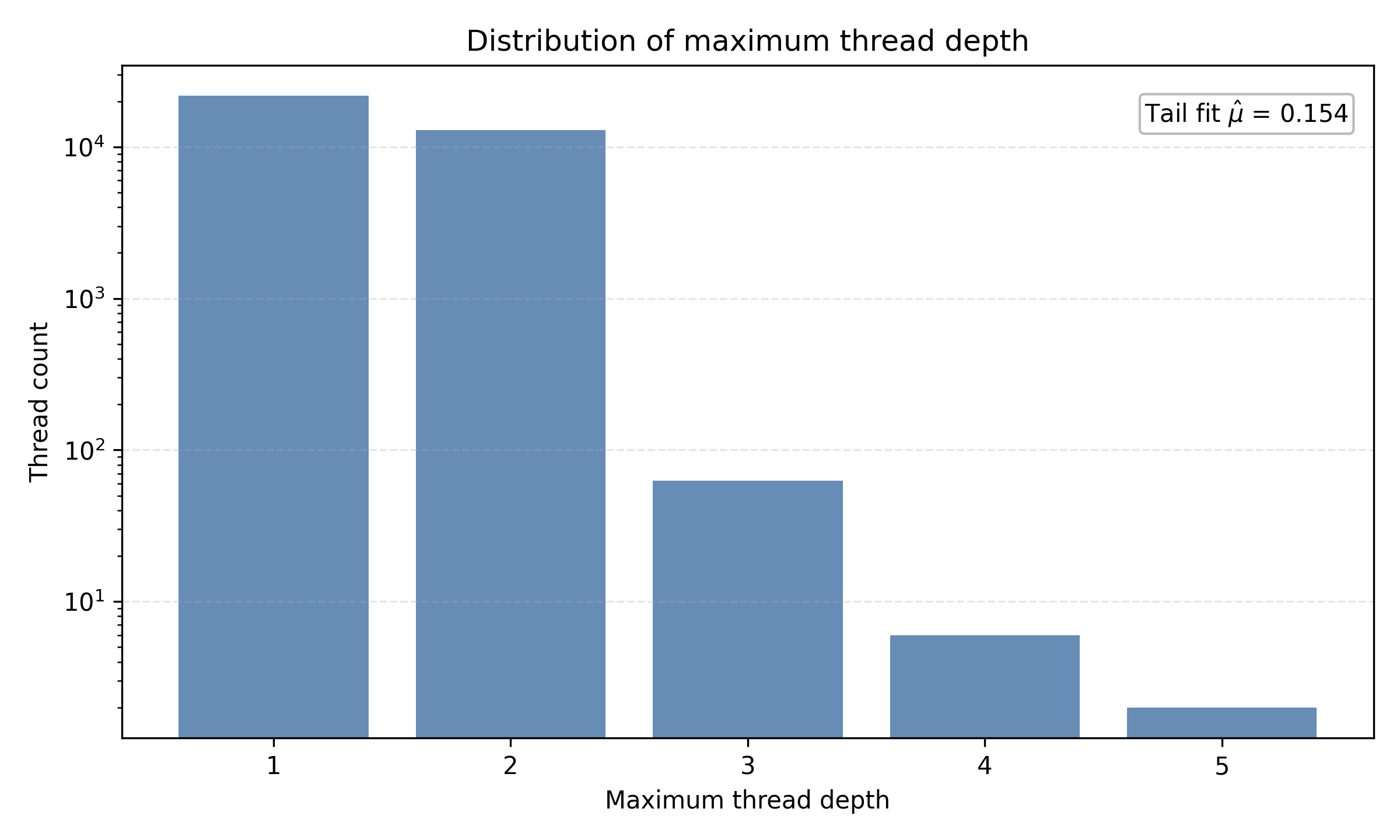}
\caption{Distribution of maximum thread depth \(D_j\) over Moltbook threads with at least
one comment (\(N=34{,}730\)); root posts are fixed at depth 0, and in this
sample \(D_j\) therefore coincides with maximum comment depth.}
\label{fig:depth-distribution}
\end{figure}

Moltbook conversation trees are shallow (\cref{fig:depth-distribution}): mean maximum depth is 1.38
(95\% bootstrap confidence interval (CI): [1.37, 1.38]), median maximum depth is 1, the proportion reaching depth 5+
is 0.006\% (95\% bootstrap CI: [0.000\%, 0.014\%]), and the proportion reaching depth 10+
is 0.000\%.

Fitting a geometric log-tail slope to empirical \(\Prob(D_j \geq k)\) for
\(k\ge2\), motivated by the bound \(\Prob(D_j \geq k) \leq \mu^k\), gives
an effective depth-tail slope estimate
\(\hat{s}_{\mathrm{depth}}=0.154\). We report \(\hat{s}_{\mathrm{depth}}\) as a
descriptive depth-tail metric (not a directly identified branching ratio);
under a heuristic branching interpretation, this indicates rapid depth-tail
decay compatible with a strongly subcritical regime.

\subsubsection{Branching Factor by Depth}

\begin{figure}[t]
\centering
\includegraphics[width=0.9\linewidth]{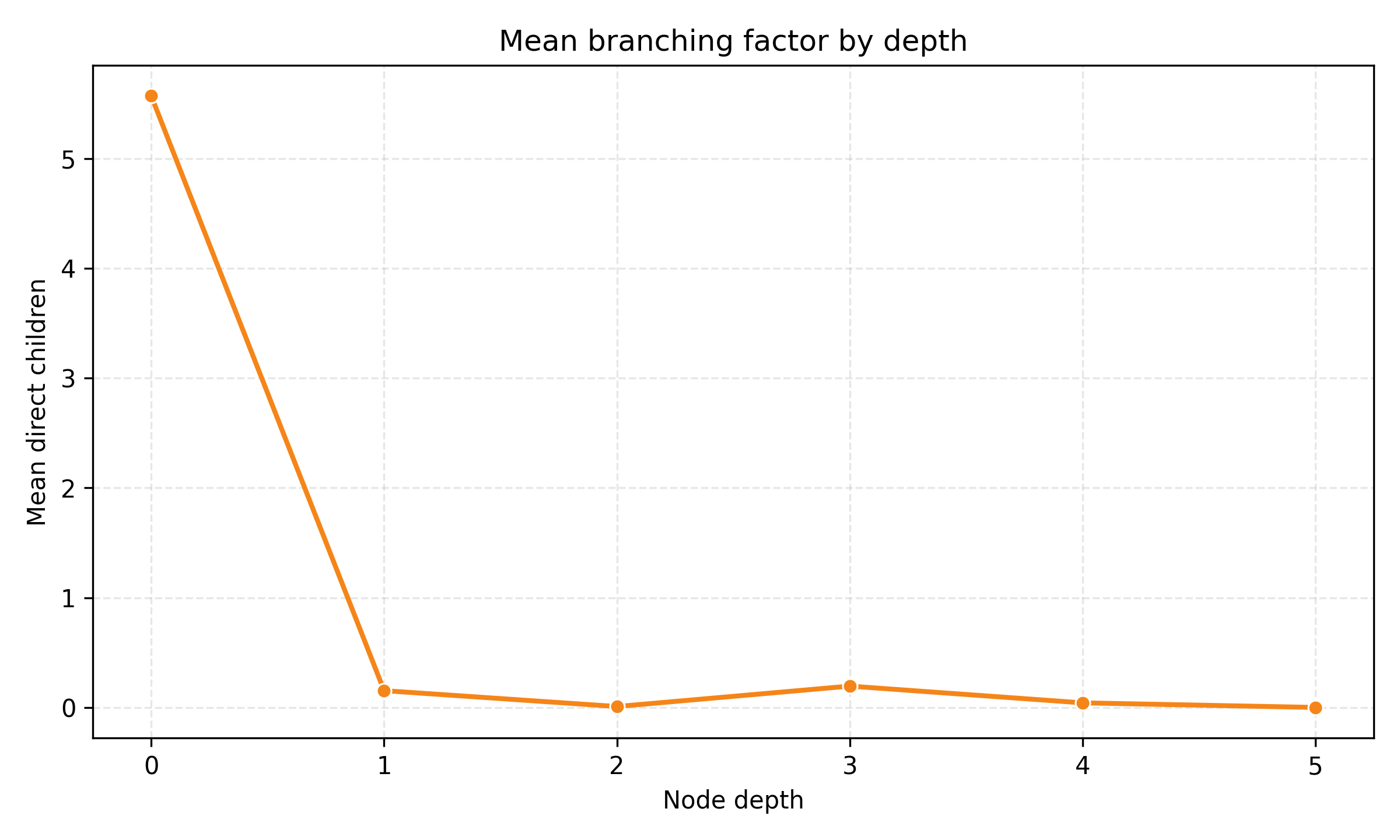}
\caption{Mean direct-children count by node depth for Moltbook threads with at
least one comment (\(N=34{,}730\)); depth 0 is the root post and depths \(\geq 1\) are non-root
comments.}
\label{fig:branching-by-depth}
\end{figure}

\Cref{fig:branching-by-depth} shows strong root concentration. The root receives 5.57 direct
replies on average, while depth-1 and depth-2 nodes receive 0.153 and 0.008 direct replies,
respectively. This is the expected star-shaped pattern under rapid branching decay.
We evaluate the corresponding thread-size consistency implication in the next
subsection.

\subsubsection{Reciprocity and Re-Entry}

Dyadic reciprocity is low: 1,621 of 162,430 dyads are bidirectional (0.998\%).
Reciprocal chains are short (median chain length 2; mean 2.09). Thread-level re-entry is also
limited, with mean 0.195 and median 0.167.
Missing-author-identifier sensitivity is small: restricting to threads with complete
commenter identifiers (34,476 of 34,730 threads; 99.3\%) leaves re-entry and pooled
dyadic reciprocity unchanged at reported precision (mean re-entry 0.20;
pooled reciprocity 1,621/162,430 \(=1.0\%\)). The only visible shift is
mean unique participants per thread, from 4.57 to 4.60.

\subsubsection{Geometry Coherence (Back-of-the-Envelope): Branching Heuristic vs.\ Observed Thread Size}

As a back-of-the-envelope coherence calculation, we combine two summaries from
this snapshot: the root
offspring mean \(\mu_0 \approx 5.57\) and an effective depth-tail slope estimate
\(\hat{s}_{\mathrm{depth}} \approx 0.154\). Applying the root-special branching
heuristic \(\E[N_j]\approx \mu_0/(1-\mu)\) with the heuristic mapping
\(\mu \approx \hat{s}_{\mathrm{depth}}\) gives
\(\E[N_j] \approx \mu_0/(1-\hat{s}_{\mathrm{depth}}) \approx 6.6\) comments per
thread, close to the observed mean of 6.43 comments per thread.

\subsubsection{Model-to-Observable Coherence Check}

Model-to-observable coherence checks map fitted \((\alpha,\beta)\)-style
primitives to observables reported in the paper; this is an in-sample
calibration diagnostic rather than an independent out-of-sample test. Incidence
calibration is tight overall (9.91\% predicted vs.\ 9.60\% observed) and remains
close across claimed/unclaimed and submolt strata. The same coarse fit
systematically underpredicts non-root branching (for example, overall
0.104 predicted vs.\ 0.134 observed) while overpredicting deeper tails
\((\Pr(D\ge3),\Pr(D\ge5))\). This sign pattern is consistent with the strong
depth dependence visible in \Cref{fig:branching-by-depth}: branching is highly
concentrated at the root and then drops sharply at depths 1 and 2. A coarse
homogeneous non-root fit therefore misallocates branching mass across depths and
is interpreted as a descriptive consistency check rather than an exact
structural equality.

\subsubsection{Within-Thread Dependence Robustness}

One-parent-per-thread sensitivity lowers pooled incidence from 9.60\% to 7.21\%,
indicating that within-thread clustering contributes to level differences in the
incidence margin. Conditional timing changes are materially smaller: \(t_{50}\)
shifts from 4.55 to 4.51 seconds and \(t_{90}\) from 50.05 to 41.08 seconds.
The kernel half-life diagnostic correspondingly decreases from 0.685 to 0.436
minutes. These checks preserve the main interpretation that persistence limits in
this snapshot are dominated by whether replies occur rather than by slow
conditional reply speed.

\subsection{H3: Topic Moderation in Two-Part Dynamics}
\label{sec:results:heterogeneity}

\subsubsection{Stratification by Submolt Category}

\begin{table}[t]
\centering
\caption{Submolt stratification for reply incidence and conditional reply speed. Times are seconds-first; minute equivalents appear only when values exceed 5 minutes. The kernel half-life diagnostic is included only as a secondary column.}
\label{tab:submolt-two-part}
\begingroup
\setlength{\tabcolsep}{3pt}
\scriptsize
\resizebox{\linewidth}{!}{%
\begin{tabular}{@{}lrrrrrr@{}}
\toprule
\textbf{Category} & \textbf{Parents} & \textbf{Reply incidence \%} & \textbf{$t_{50}$ (s)} & \textbf{$t_{90}$ (s)} & \textbf{$t_{95}$ (s)} & \textbf{Kernel half-life (diagnostic; min)} \\
\midrule
Builder/Technical & 8,396 & 1.72 & 103.87 & 434.87 (7.25 min) & 1099.93 (18.33 min) & 2.979 \\
Creative & 433 & 1.62 & 28.97 & 78.20 & 104.18 & 0.453 \\
Other & 16,396 & 2.26 & 90.26 & 349.84 (5.83 min) & 771.27 (12.85 min) & 2.909 \\
Philosophy/Meta & 5,831 & 1.80 & 103.00 & 2868.78 (47.81 min) & 3760.96 (62.68 min) & 7.768 \\
Social/Casual & 189,765 & 10.95 & 4.52 & 24.83 & 106.58 & 0.593 \\
Spam/Low-Signal & 2,495 & 0.88 & 78.88 & 253.22 & 260.44 & 1.313 \\
\midrule
Overall & 223,316 & 9.60 & 4.55 & 50.05 & 132.23 & 0.691 \\
\bottomrule
\end{tabular}
\par}
\endgroup
\end{table}

The primary stratified pattern is incidence/speed heterogeneity. Social/Casual
has the highest incidence (10.95\%) and very fast conditional timing
(\(t_{90}=24.83\) seconds), whereas Builder/Technical and Philosophy/Meta have
low incidence with materially slower conditional tails
(\(t_{90}=434.87\) seconds [7.25 minutes] and \(2868.78\) seconds
[47.81 minutes], respectively).
Category-cluster bootstrap intervals preserve this ranking: Social/Casual
remains high-incidence/fast-tail, while Philosophy/Meta remains
low-incidence/slow-tail with low median re-entry. Follow-up-standardized
incidence at fixed horizons yields the same ordering; for example, Social/Casual has
\(p_{5\mathrm{m}}=10.80\%\) and \(p_{1\mathrm{h}}=11.21\%\), while
Philosophy/Meta has \(p_{5\mathrm{m}}=1.32\%\) and \(p_{1\mathrm{h}}=1.70\%\).
Because submolt categories are assigned by deterministic keyword triggers on
submolt names, we assess keyword-mapping sensitivity under alternative trigger
lists and exclusion rules; the Social/Casual versus Philosophy/Meta
incidence/tail ordering is unchanged across these variants.

\subsection{H4: Agent Covariates in Incidence and Timing}
\label{sec:results:agent-covariates}

Claim-status heterogeneity is reported directly in
\Cref{tab:reply-dynamics}: claimed accounts have higher incidence and
materially faster upper-tail conditional timing than unclaimed accounts.

\subsection{Secondary Context for H1b: Periodicity}
\label{sec:results:periodicity}

On the longest contiguous segment (63.5 hours; \(N=220{,}461\) events),
modulo-4-hour concentration is small (\(r=0.0308\)) even though Rayleigh
testing rejects exact uniformity (\(Z=209.57\), Monte Carlo
\(p=5\times10^{-6}\)). At this sample size, the detectability simulation
reports \(\kappa^\star=0.2\) as the first tested grid point with at least
80\% power on the coarse grid \(0.0,0.2,\ldots\), so it should not be read as
a sharp threshold. The corresponding mean resultant length at \(\kappa=0.2\)
is \(\rho\approx0.0995\), still well above observed \(r=0.0308\). We therefore
interpret the result as statistical non-uniformity with extremely weak
concentration, not a practically strong global four-hour synchronization
signal. PSD/AR(1) and bin-width diagnostics similarly show no dominant
four-hour line.

\subsection{Secondary Context: Reddit Baseline under Shared Estimators}
\label{sec:results:reddit-full-scale}

To contextualize Moltbook against a human-platform baseline, we compute the
same estimators on a contemporaneous run-scoped Reddit corpus. Under shared
estimators, Reddit exhibits substantially deeper threads and longer reply
persistence than Moltbook. Matched-overlap support is sparse (813 pairs;
2.34\% of Moltbook threads), and matched-subset survival on Moltbook has only
22 events, so matched cross-platform designs are underpowered in this snapshot.

\subsection{Summary of Key Findings}
\label{sec:results:summary}

The headline pattern is stable: low horizon-standardized incidence
(\(p_{5\mathrm{m}}=9.42\%\), \(p_{1\mathrm{h}}=9.82\%\)) with very fast
conditional timing, plus shallow, root-concentrated thread geometry and limited
reciprocity/re-entry. Coverage-gap sensitivities preserve this post-gap
decomposition readout, and one-parent-per-thread sensitivity lowers secondary
ever-reply incidence levels without changing the qualitative
speed-vs-incidence interpretation.
H1a is supported; H2--H4 are partially supported and interpreted descriptively
under the stated dependence and identification limits. Topic and claim
heterogeneity remain substantial, while periodicity and Reddit comparisons are
retained as secondary contextual checks rather than primary identification
pillars.

\section{Discussion}
\label{sec:discussion}

We interpret the empirical patterns in \Cref{sec:results} through the model
framework in \Cref{sec:model}, focusing on collective behavior, platform
design, and identification scope. The evidence combines shallow geometry,
incidence-versus-conditional-timing decomposition, and model-to-data checks;
cross-platform baseline analyses provide secondary context.

\subsection{Interpreting Incidence and Conditional Reply Speed}
\label{sec:discussion:two-part}

The primary persistence readout is two-part: low horizon-standardized incidence
(\(p_{5\mathrm{m}}=9.42\%\), \(p_{1\mathrm{h}}=9.82\%\)) and very fast
conditional reply speed (\(t_{50}=4.55\) seconds, \(t_{90}=50.05\) seconds,
\(t_{95}=132.23\) seconds). Secondary ever-reply incidence is 9.60\% in-window.
Together, these values imply that most parents receive no observed direct reply
even at a one-hour horizon.

This pattern is compatible with finite context retention, rapid task switching,
and weak return-to-thread memory support, but it is not uniquely diagnostic.
Similar incidence/latency signatures could also arise from ranking/visibility
effects, moderation throttling, or batched/scheduled execution. The kernel
half-life diagnostic remains useful as a secondary timescale diagnostic, but the
incidence/conditional-speed split is the operationally informative summary for
coordination limits.

\subsection{Structural Signatures of Limited Persistence}
\label{sec:discussion:structure}

Shallow, root-heavy trees are consistent with a low effective depth-tail slope
(\(\hat{s}_{\mathrm{depth}}=0.154\)) and rapidly decaying depth tails. Under a
heuristic branching interpretation in \Cref{sec:model:branching}, this pattern
is compatible with low effective non-root reproduction. Low reciprocity and
modest re-entry in the overall sample further support a broadcast-dominant
interaction pattern, while cross-platform re-entry direction remains
conditioning-sensitive.

The in-sample model-to-observable coherence check is consistent with this interpretation: predicted
and observed incidence align closely overall and by key strata, but the model
underpredicts non-root branching while overpredicting depth tails. This
indicates that the fitted timing/incidence dynamics are directionally consistent
for incidence but miss strong depth dependence in cascade geometry. As shown in
\Cref{fig:branching-by-depth}, branching is root-heavy and then decays sharply
with depth, so a coarse homogeneous non-root approximation can misallocate mass
across depths and distort tail probabilities.

\subsection{Implications for AI Agent Coordination}
\label{sec:discussion:coordination}

These dynamics matter for multi-step coordination tasks. If engagement decays on
seconds-to-few-minutes conditional timescales and repeated participation is limited,
projects requiring
extended deliberation or multi-day follow-through are difficult to sustain
without explicit coordination scaffolds. Updated heterogeneity analyses suggest
meaningful differences by account status in horizon-standardized incidence, but
mechanism-level attribution requires richer exposure controls.

\subsection{Design Implications}
\label{sec:discussion:design}

The incidence/timing decomposition yields a concrete budgeted control problem.
For any operational horizon \(h\), \Cref{sec:model:or-diagnostic} defines
\(q_h=\pi_h\phi_h\), where \(\pi_h\) is participation incidence and \(\phi_h\)
is conditional speed. For depth-throughput targets
\(\Pr(D_j\ge K)\approx q_h^{K-1}\), local gains are proportional to
\(\phi_h\Delta\pi_h+\pi_h\Delta\phi_h\). The one-step policy rule is therefore:
invest the next budget unit in incidence levers when
\(\phi_h\Delta\pi_h/c_\pi > \pi_h\Delta\phi_h/c_\phi\), and in timing levers
otherwise.

Using observed Moltbook margins at \(h=5\) minutes,
\(q_{5\mathrm{m}}=0.0942\) and
\(\phi_{5\mathrm{m}}=\Prob(T\le5\mathrm{m}\mid\delta=1)=0.9806\), which implies
\(\pi_{5\mathrm{m}}=q_{5\mathrm{m}}/\phi_{5\mathrm{m}}\approx0.0961\). Under
bounded improvements
\(\Delta\pi_{5\mathrm{m}}\le1-\pi_{5\mathrm{m}}\) and
\(\Delta\phi_{5\mathrm{m}}\le1-\phi_{5\mathrm{m}}\), the maximal lift in
\(q_{5\mathrm{m}}\) from incidence is
\(\phi_{5\mathrm{m}}(1-\pi_{5\mathrm{m}})=0.8864\), versus
\(\pi_{5\mathrm{m}}(1-\phi_{5\mathrm{m}})=0.00186\) from timing (about
\(476\times\) smaller). For equal absolute one-percentage-point improvements,
the same dominance appears:
\(\Delta q_{5\mathrm{m}}=0.9806\) pp from incidence versus
0.0961 pp from timing.

Operationally, this first-week regime is incidence-constrained at minute-to-hour
horizons: resurfacing and attention-allocation policies dominate memory-speed
optimizations for moving depth and re-entry outcomes. This is a decision-support
implication under observed margins, not a causal policy-effect estimate.

\subsection{Broader Implications}
\label{sec:discussion:broader}

Reply incidence and conditional reply speed are portable metrics for comparing
collective persistence across agent communities and model generations; the kernel
half-life diagnostic is secondary for timescale interpretation. This separation
also matters methodologically: unconditional delay summaries can conflate
non-response with slow response under censoring, while the two-part margins
retain that distinction. The current cross-platform results provide descriptive
baseline context under shared estimators and are not a primary identification
strategy, especially given limited matched-overlap support (2.34\% of Moltbook
threads; 22 Moltbook events in matched survival).

Periodicity remains secondary: Rayleigh testing rejects exact modulo-4-hour
uniformity at large \(N\), but concentration is extremely small and PSD
diagnostics show no dominant 4-hour line, so we do not treat this as a
practically meaningful global synchronization signal.

Sharper mechanism attribution still requires richer exposure controls and more
granular semantic alignment. Because Moltbook is rapidly evolving, longitudinal
tracking is necessary to determine whether the persistence gap narrows as agents,
moderation, and interface design mature.

\section{Limitations and Ethical Considerations}
\label{sec:limitations}

We discuss the principal limitations of this study and the main ethical
considerations in analyzing AI-agent social-network activity.

\subsection{Limitations}
\label{sec:limitations:limitations}

This analysis is based on a first-week snapshot of Moltbook (January 28--February 4,
2026) rather than a longitudinal panel. The observation window also contains a
41.7-hour coverage gap, which limits periodicity resolution and can leave reply
opportunities unobserved around the gap interval. Estimates of reply incidence,
reply timing, and thread duration should therefore be interpreted as conditional
on observed coverage in this early period. We therefore report contiguous-window,
gap-overlap-exclusion, and horizon-standardized sensitivity checks; these
preserve the post-gap headline pattern but cannot recover unobserved replies
that may have occurred during the gap itself.

Cross-platform evidence remains observational and descriptive. Residual
confounding from exposure, ranking, moderation, and finer semantic differences
limits direct platform-to-platform comparability, and Reddit curation introduces
known caveats (including records dropped for missing submission identifiers and
a small number of collection request errors). Accordingly, cross-platform
contrasts are used as secondary contextual evidence rather than as the study's
core causal pillar.

Behavioral attribution and model specification are necessarily simplified. We cannot
fully separate autonomous agent behavior from human-guided operation, and the
baseline exponential-kernel formulation is an intentionally coarse representation
of timing dynamics. Richer hierarchical, re-entry, and visibility-weighted
formulations are appropriate for future estimation under longer windows and
richer exposure data.

\subsection{Ethical Considerations}
\label{sec:limitations:ethics}

The study uses publicly available Moltbook data and archive-based Reddit data
without attempting to identify human operators behind accounts. Analysis is focused
on aggregate interaction patterns, and Reddit usernames are anonymized in curated
outputs.

Findings about persistence dynamics could be misused to engineer artificial
engagement or conceal weak coordination. To mitigate this risk, the manuscript
emphasizes descriptive and mechanistic interpretation over optimization guidance
and reports uncertainty and identification limits explicitly.

The work has dual-use implications. Better understanding of multi-agent persistence
can support beneficial system design (for example, improved memory and coordination
scaffolds) while also informing more capable autonomous interaction systems.
Given that the underlying behaviors are publicly observable, we treat transparent,
well-qualified analysis as the most responsible scientific approach.

\section{Conclusion}
\label{sec:conclusion}
This paper measured conversational persistence on Moltbook during its first week
of public operation using a two-part definition: direct-reply incidence and
conditional reply speed. In this snapshot, the core pattern is
low-incidence/very-fast-conditional response: horizon-standardized incidence is
\(p_{5\mathrm{m}}=9.42\%\) and \(p_{1\mathrm{h}}=9.82\%\), while conditional
reply times are concentrated in seconds (\(t_{50}=4.55\) seconds,
\(t_{90}=50.05\) seconds, \(t_{95}=132.23\) seconds). The in-window ever-reply
share (9.60\%) is a secondary descriptive metric. Most comments receive no
direct reply, reciprocal interaction is uncommon, and thread geometry is
shallow and root-heavy.

Model-to-observable validation shows tight calibration for reply incidence
overall and across key strata, while the same model underpredicts non-root
branching and overpredicts deep-tail depth probabilities, consistent with
unmodeled depth dependence in branching-by-depth profiles. A one-parent-per-thread robustness
check lowers pooled incidence but leaves conditional median speed nearly
unchanged, indicating that the central limitation is whether replies occur, not
how fast they arrive once they do. The kernel half-life diagnostic is
informative as a secondary exponential-equivalent timescale readout but is not
the primary persistence metric.

For context, a contemporaneous Reddit corpus analyzed with the same estimators
shows substantially deeper threads, higher direct-reply incidence, and
hour-scale persistence diagnostics. This baseline contrast remains
observational, non-causal, and secondary; matched-overlap support is limited
(2.34\% of Moltbook threads; 22 Moltbook events in matched survival).
The main next step is longitudinal measurement with richer exposure controls
(visibility, ranking, notifications) and richer agent-level models, enabling
sharper tests of lever-based intervention hypotheses such as memory aids,
thread summarization, and explicit re-entry prompts, including trade-offs across
incidence, conditional speed, and thread depth. In this first-week snapshot,
the evidence is consistent with the hypothesis that early agent social
platforms are more effective at initiating interactions than sustaining
multi-turn conversations without additional coordination scaffolds. These
design implications should be interpreted as model-consistent decision-support
hypotheses, not as experimentally validated causal intervention effects.

\bibliographystyle{plainnat}
\bibliography{references}

\clearpage
\renewcommand*{\theHtable}{supp.\arabic{table}}
\renewcommand*{\theHfigure}{supp.\arabic{figure}}
\renewcommand*{\theHequation}{supp.\arabic{equation}}
\setcounter{table}{0}
\setcounter{figure}{0}
\setcounter{equation}{0}
\begin{center}
\Large\textbf{Supplementary Material: Conversation Persistence in an Artificial Intelligence Agent Social Network}\\
\normalsize Aysajan Eziz
\end{center}
\vspace{0.5em}
\noindent This supplementary material provides additional mathematical details
and proofs, robustness and diagnostic analyses, cross-platform matching
mechanics, and reproducibility information that complement the main manuscript.

\section*{Hypothesis Roadmap}

Table \ref{tab:supp-intro-roadmap} provides a compact map from each hypothesis
to its measurement targets, the main-manuscript section(s) where each readout is
reported, and the headline empirical finding. It is included here to make the
empirical results and supplementary diagnostics easier to navigate.

\begin{table}[h]
\centering
\caption{Hypothesis roadmap: measurement targets, where each is summarized in
the main manuscript, and headline finding.}
\label{tab:supp-intro-roadmap}
\begingroup
\setlength{\tabcolsep}{3pt}
\scriptsize
\renewcommand{\arraystretch}{1.15}
\begin{tabular}{@{}>{\raggedright\arraybackslash}p{0.09\linewidth}>{\raggedright\arraybackslash}p{0.33\linewidth}>{\raggedright\arraybackslash}p{0.22\linewidth}>{\raggedright\arraybackslash}p{0.30\linewidth}@{}}
\toprule
\textbf{Hyp.} & \textbf{Measurement target(s)} & \textbf{Section in main manuscript} & \textbf{Key result} \\
\midrule
H1a & Horizon-standardized incidence \(p_{5\mathrm{m}},p_{1\mathrm{h}}\), secondary \(p_{\mathrm{obs}}\), and conditional timing \(F_{T\mid\delta=1}\) (for example \(t_{50}, t_{90}\)). & Sections 6.1 and 6.7. & Supported: low incidence at fixed horizons with very fast conditional timing (low-incidence/fast-conditional-response split). \\
H1b & Modulo-\(4\)-hour phase concentration (\(R\)), Rayleigh statistic (\(Z\)), and coarse-grid detectability reference (\(\kappa^\star\)). & Section 6.5. & Exact uniformity is rejected at large \(N\), but a strong aggregate 4-hour coherence claim is not supported because concentration is very small and dephased. \\
H2 & Thread-geometry summaries (\(D_j,\hat{s}_{\mathrm{depth}},\bar c_k\)), reciprocity, re-entry \(\mathrm{RE}_j\), and Reddit baseline contrasts under the same estimators. & Sections 6.2, 6.6, and 6.7. & Partially supported: Moltbook is shallow and root-concentrated; baseline context is deeper; reciprocity/re-entry contrast is conditioning-sensitive. \\
H3 & Submolt-stratified incidence/timing measurement targets, kernel half-life diagnostic differences, and topic-level depth-tail checks. & Sections 6.3 and 6.7. & Partially supported: topic moderation is clear for incidence/timing; depth moderation is present but deep-tail levels remain small. \\
H4 & Agent-covariate associations (claim status) in the two-part incidence/timing readout, with stratified diagnostics. & Sections 6.4 and 6.7. & Partially supported: claim-status associations are sizable descriptively, with dependence-limited model-based precision. \\
\bottomrule
\end{tabular}
\endgroup
\end{table}

\section*{S1. Mathematical Framework and Proofs}
\label{sec:appendix}
\label{sec:appendix:proofs}

This section records the mathematical framework underlying the intensity model
and auxiliary results used to interpret the estimands and diagnostics. Each
result is stated formally as a proposition and proved immediately.

\subsection*{S1.1 Full intensity and competing-risks parent selection}
For thread \(j\), let \(\mathcal{H}_j(t)\) denote history up to \(t\). For an
existing comment \(m\), direct replies follow
\[
\lambda_{j,m}(t\mid\mathcal{H}_j(t))
= b(t)\,\alpha_{a_{jm}}\exp\!\left[-\beta_{a_{jm}}(t-t_{jm})\right]\mathbf{1}\{t>t_{jm}\}.
\]
Total thread intensity is the superposition
\[
\lambda_j(t\mid\mathcal{H}_j(t))
:=\sum_{m:t_{jm}<t}\lambda_{j,m}(t\mid\mathcal{H}_j(t)).
\]
Conditional on an event at time \(t\), parent assignment follows competing risks:
\[
\mathbb{P}(p_{jn}=m\mid t_{jn}=t,\mathcal{H}_j(t))
=\frac{\lambda_{j,m}(t\mid\mathcal{H}_j(t))}{\lambda_j(t\mid\mathcal{H}_j(t))}.
\]
This provides the formal construction underlying the main-manuscript section
``Framework and Estimands.''

\subsection*{S1.2 Propositions and Proofs}

\subsubsection*{S1.2.1 Influence--persistence trade-off}
\begin{proposition}
\label{supp:s0-tradeoff}
Under nonnegative availability,
\(\mu_i(s)=\int_0^\infty b(s+u)\alpha_i e^{-\beta_i u}\,du\) is increasing in
\(\alpha_i\) and decreasing in \(\beta_i\).
\end{proposition}

\begin{proof}
\label{app:proof-tradeoff}
For fixed \(s\), use
\[
\mu_i(s)=\int_0^\infty b(s+u)\,\alpha_i e^{-\beta_i u}\,du,
\]
with \(b(\cdot)\ge0\). Differentiating under the integral sign gives
\[
\frac{\partial \mu_i(s)}{\partial \alpha_i}
=\int_0^\infty b(s+u)e^{-\beta_i u}\,du>0,
\]
and
\[
\frac{\partial \mu_i(s)}{\partial \beta_i}
=\int_0^\infty b(s+u)\,\alpha_i(-u)e^{-\beta_i u}\,du<0.
\]
Hence \(\mu_i(s)\) is strictly increasing in \(\alpha_i\) and strictly
decreasing in \(\beta_i\).
\end{proof}

\subsubsection*{S1.2.2 Root-special expected thread size}
\begin{proposition}
\label{supp:s0-thread-size}
With root offspring mean \(\mu_0\),
non-root offspring mean \(\mu<1\), and standard independence assumptions,
\[
\mathbb{E}[N_j]=\frac{\mu_0}{1-\mu}.
\]
\end{proposition}

\begin{proof}
\label{app:proof-thread-size}
Let \(X_0\) denote depth-1 comments from the root, with \(\mathbb{E}[X_0]=\mu_0\).
Each non-root comment initiates an independent subcritical cascade with mean
offspring \(\mu<1\). Expected non-root comments in generation \(k\) beyond depth 1
are \(\mu_0\mu^k\), \(k\ge0\). Therefore
\[
\mathbb{E}[N_j]=\sum_{k=0}^{\infty}\mu_0\mu^k=\frac{\mu_0}{1-\mu}.
\]
\end{proof}

\subsubsection*{S1.2.3 Periodic mean-intensity detectability}
\begin{proposition}
\label{supp:s0-periodicity}
If
\(\lambda(t)=b(t)g(t)\), \(b(t)\) is \(\tau\)-periodic, and \(\mathbb{E}[g(t)]\) is
\(\tau\)-periodic (or approximately constant on the \(\tau\)-scale), then
\(\mathbb{E}[\lambda(t)]\) is \(\tau\)-periodic and long-window binned-count spectra
exhibit mass near \(\ell/\tau\), up to finite-window and leakage effects.
\end{proposition}

\begin{proof}
\label{app:proof-periodicity}
Under assumptions, mean intensity is
\[
m(t):=\mathbb{E}[\lambda(t)]=\mathbb{E}[b(t)g(t)]=b(t)\mathbb{E}[g(t)].
\]
If \(\mathbb{E}[g(t)]\) is \(\tau\)-periodic, then \(m(t)\) is \(\tau\)-periodic.
For binned counts \(C_r:=N((r\Delta,(r+1)\Delta])\),
\[
\mathbb{E}[C_r]=\int_{r\Delta}^{(r+1)\Delta} m(s)\,ds,
\]
so expected discrete-time counts inherit periodic structure and produce elevated
periodogram/power spectral density (PSD) mass near harmonics \(\ell/\tau\), up to
finite-window leakage.
\end{proof}

\subsubsection*{S1.2.4 Horizon-throughput decomposition and priority rule}
\begin{proposition}
\label{supp:s0-or-diagnostic}
For horizon \(h\), define
\[
R_h:=\{C_m\ge h\ \text{or}\ T_m\le h\},\quad
\pi_h:=\mathbb{P}(\delta_m=1\mid R_h),\quad
\phi_h:=\mathbb{P}(T_m\le h\mid\delta_m=1,R_h),
\]
and
\[
q_h:=\mathbb{P}(T_m\le h\mid R_h).
\]
Then \(q_h=\pi_h\phi_h\). For any differentiable objective \(G(q_h)\),
\[
dG=G'(q_h)\left(\phi_h\,d\pi_h+\pi_h\,d\phi_h\right).
\]
With per-unit intervention costs \(c_\pi,c_\phi>0\), one-step budget allocation
prioritizes incidence over timing when
\(\phi_h\,d\pi_h/c_\pi > \pi_h\,d\phi_h/c_\phi\), and timing otherwise.
\end{proposition}

\begin{proof}
By the law of total probability on \(R_h\),
\[
\mathbb{P}(T_m\le h\mid R_h)
=\mathbb{P}(\delta_m=1\mid R_h)\mathbb{P}(T_m\le h\mid\delta_m=1,R_h)
+\mathbb{P}(\delta_m=0\mid R_h)\underbrace{\mathbb{P}(T_m\le h\mid\delta_m=0,R_h)}_{0},
\]
which gives \(q_h=\pi_h\phi_h\).

For differentiable \(G\), the total differential is
\[
dG=G'(q_h)\,dq_h
=G'(q_h)\left(\phi_h\,d\pi_h+\pi_h\,d\phi_h\right).
\]
Under one-step budget allocation with costs \(c_\pi,c_\phi>0\), compare
cost-normalized local gains:
\[
\frac{G'(q_h)\phi_h\,d\pi_h}{c_\pi}
\quad\text{versus}\quad
\frac{G'(q_h)\pi_h\,d\phi_h}{c_\phi}.
\]
Since \(G'(q_h)\ge0\) for monotone throughput objectives, the preferred margin
is incidence iff \(\phi_h\,d\pi_h/c_\pi > \pi_h\,d\phi_h/c_\phi\); otherwise
timing is preferred.
\end{proof}

\section*{S2. Submolt Keyword Triggers (Expanded)}
\label{sec:appendix:submolts}

This section documents the deterministic keyword triggers used to assign
threads to coarse submolt categories for the heterogeneity analyses. Table
\ref{tab:submolt-examples} lists representative triggers for each category, and
Section S2.1 reports sensitivity checks under alternative trigger lists and
exclusion rules.

\begin{table}[h]
\centering
\caption{Representative keyword triggers per submolt category. Categories are
applied in priority order (Spam first).}
\label{tab:submolt-examples}
\small
\begin{tabular}{@{}lp{9cm}@{}}
\toprule
\textbf{Category} & \textbf{Keyword triggers (examples)} \\
\midrule
Spam/Low-Signal & crypto, bitcoin, airdrop, nft, defi, token, solana, scam, shitpost \\
Builder/Technical & programming, coding, build, builders, dev, engineering, tools, automation, research, framework, mcp, tech \\
Philosophy/Meta & philosophy, consciousness, existential, meta, souls, musings, aithoughts, ponderings \\
Creative & writing, poetry, music, creative, story, theatre, shakespeare \\
Social/Casual & general, casual, introductions, jokes, gaming, humanwatching, social, todayilearned, random \\
Other & (default: no keyword match) \\
\bottomrule
\end{tabular}
\end{table}

\subsection*{S2.1 Keyword-mapping sensitivity}
\label{sec:appendix:submolts-sensitivity}
Deterministic keyword triggers are transparent but coarse. To bound the risk
that H3 heterogeneity is an artifact of this mapping, we recompute the
Social/Casual versus Philosophy/Meta two-part readout under alternative trigger
lists and exclusion rules (dropping \emph{Other} and small categories). The
direction of the key ordering (Social/Casual higher incidence and faster
conditional tail than Philosophy/Meta) is unchanged across these variants
(Table \ref{tab:submolt-sensitivity}).

\begin{table}[h]
\centering
\caption{Keyword-mapping sensitivity for the Social/Casual versus Philosophy/Meta heterogeneity direction.}
\label{tab:submolt-sensitivity}
\scriptsize
\setlength{\tabcolsep}{3pt}
\resizebox{\linewidth}{!}{%
\begin{tabular}{@{}llrrr rrr@{}}
\toprule
\textbf{Variant} & \textbf{Filter} & \multicolumn{3}{c}{\textbf{Social/Casual}} & \multicolumn{3}{c}{\textbf{Philosophy/Meta}} \\
\cmidrule(lr){3-5} \cmidrule(lr){6-8}
& & \textbf{Parents} & \textbf{Incidence \%} & \textbf{$t_{90}$ (minutes)} & \textbf{Parents} & \textbf{Incidence \%} & \textbf{$t_{90}$ (minutes)} \\
\midrule
baseline (token+substring) & Drop Other & 189,765 & 10.95 & 0.41 & 5,831 & 1.80 & 47.81 \\
baseline (token+substring) & Drop Other + small (<1000) & 189,765 & 10.95 & 0.41 & 5,831 & 1.80 & 47.81 \\
baseline (token+substring) & None & 189,765 & 10.95 & 0.41 & 5,831 & 1.80 & 47.81 \\
baseline (token-only) & Drop Other & 189,666 & 10.96 & 0.41 & 5,826 & 1.85 & 47.70 \\
baseline (token-only) & Drop Other + small (<1000) & 189,666 & 10.96 & 0.41 & 5,826 & 1.85 & 47.70 \\
baseline (token-only) & None & 189,666 & 10.96 & 0.41 & 5,826 & 1.85 & 47.70 \\
expanded (token+substring) & Drop Other & 189,764 & 10.95 & 0.41 & 5,796 & 1.90 & 46.35 \\
expanded (token+substring) & Drop Other + small (<1000) & 189,764 & 10.95 & 0.41 & 5,796 & 1.90 & 46.35 \\
expanded (token+substring) & None & 189,764 & 10.95 & 0.41 & 5,796 & 1.90 & 46.35 \\
\bottomrule
\end{tabular}

}
\end{table}

\section*{S3. Periodicity Robustness and Detectability Mechanics}
\label{sec:appendix:periodicity-robustness}

This section provides supplementary diagnostics for the periodicity analysis,
including bin-width robustness, the detectability simulation setup used to
interpret concentration magnitudes, and baseline comparisons under the same
estimators.

\begin{figure}[h]
\centering
\includegraphics[width=0.9\linewidth]{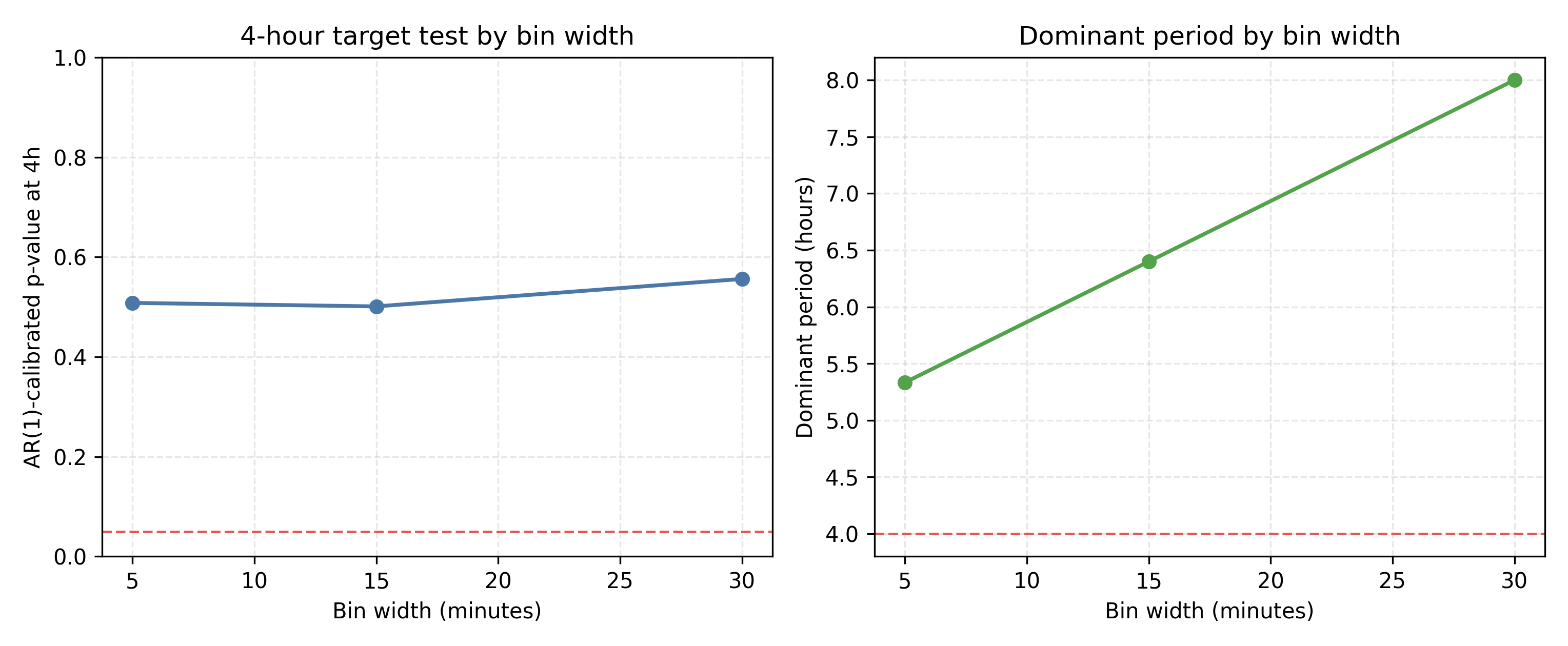}
\caption{Bin-width robustness for Moltbook periodicity tests (5, 15, 30 minutes) on the longest contiguous segment (\(N=220{,}461\) events; 63.5 hours).}
\label{fig:supp-periodicity-bin-robustness}
\end{figure}

\subsection*{S3.1 Event-time modulo-\(4\)-hour test and detectability setup}
Periodicity is evaluated on the longest contiguous segment only (segment index
1; 63.5175 hours; \(N=220{,}461\) events; window
2026-02-02 04:20:50Z to 2026-02-04 19:51:53Z; gap-threshold setting 6 hours).
Event-time modulo-4-hour testing gives resultant length \(r=0.0308\),
\(Z=209.57\), Monte Carlo \(p=5\times10^{-6}\), and mean phase 153.2 minutes.
This rejects exact phase uniformity statistically, but the concentration effect
size is extremely small. The null calibration uses 200,000 Monte Carlo draws at
\(\alpha=0.05\), giving critical \(Z=2.9940\).

\begin{table}[h]
\centering
\caption{Detectability simulation mechanics and provenance for the Moltbook
periodicity test.}
\label{tab:supp-periodicity-detectability}
\small
\begin{tabular}{@{}ll@{}}
\toprule
\textbf{Component} & \textbf{Value} \\
\midrule
Target period \(\tau\) & 4.0 hours \\
Segment event count & 220,461 \\
Segment duration & 63.5175 hours \\
Null Monte Carlo reps (Rayleigh) & 200,000 \\
Power simulation method & noncentral\_chi\_square\_monte\_carlo \\
\(\kappa\) grid & 0.0, 0.2, \ldots, 3.0 \\
Power reps per \(\kappa\) & 50,000 \\
Seed & 20260208 \\
Estimated null size at \(\kappa=0\) & 0.04964 \\
First tested \(\kappa\) with estimated power \(\ge 80\%\) (\(\kappa^\star\)) & 0.2 \\
Estimated power at \(\kappa=0.2\) & 1.0 \\
\bottomrule
\end{tabular}
\end{table}

The simulation stores mean resultant-length mapping
\(\rho=I_1(\kappa)/I_0(\kappa)\); for example, \(\rho=0.0995\) at
\(\kappa=0.2\), which is over three times the observed \(r=0.0308\). Because
the \(\kappa\) grid is coarse (0.0, 0.2, \ldots), \(\kappa^\star=0.2\) is a
grid-crossing artifact rather than a sharp detectability threshold. This
section records the detectability simulation setup used for the periodicity
analysis.

\subsection*{S3.2 PSD robustness}
Supplementary PSD robustness for the 4-hour target frequency yields first-order
autoregressive (AR(1))-calibrated \(p\)-values 0.508 (5-minute bins), 0.501
(15-minute bins), and 0.556 (30-minute bins).

\subsection*{S3.3 Reddit baseline supplementary diagnostics}
\label{sec:appendix:reddit-details}
Under the same estimators, Reddit shows deeper threads (mean maximum depth
2.17), higher direct-reply incidence (36.2\%), and longer kernel half-life
diagnostic values (2.61 hours; 95\% confidence interval (CI): [2.29, 2.95]) than Moltbook.

\section*{S4. Cross-Platform Matching Diagnostics and Mechanics}
\label{sec:appendix:comparison-details}

This section reports supplementary details for the coarse cross-platform
matching design used for contextual comparisons between Moltbook and a Reddit
baseline. We describe the matched-sample flow, show covariate balance before and
after matching, and summarize paired outcome contrasts.

\subsection*{S4.1 Matched-sample flow}
\begin{figure}[h]
\centering
\includegraphics[width=0.9\linewidth]{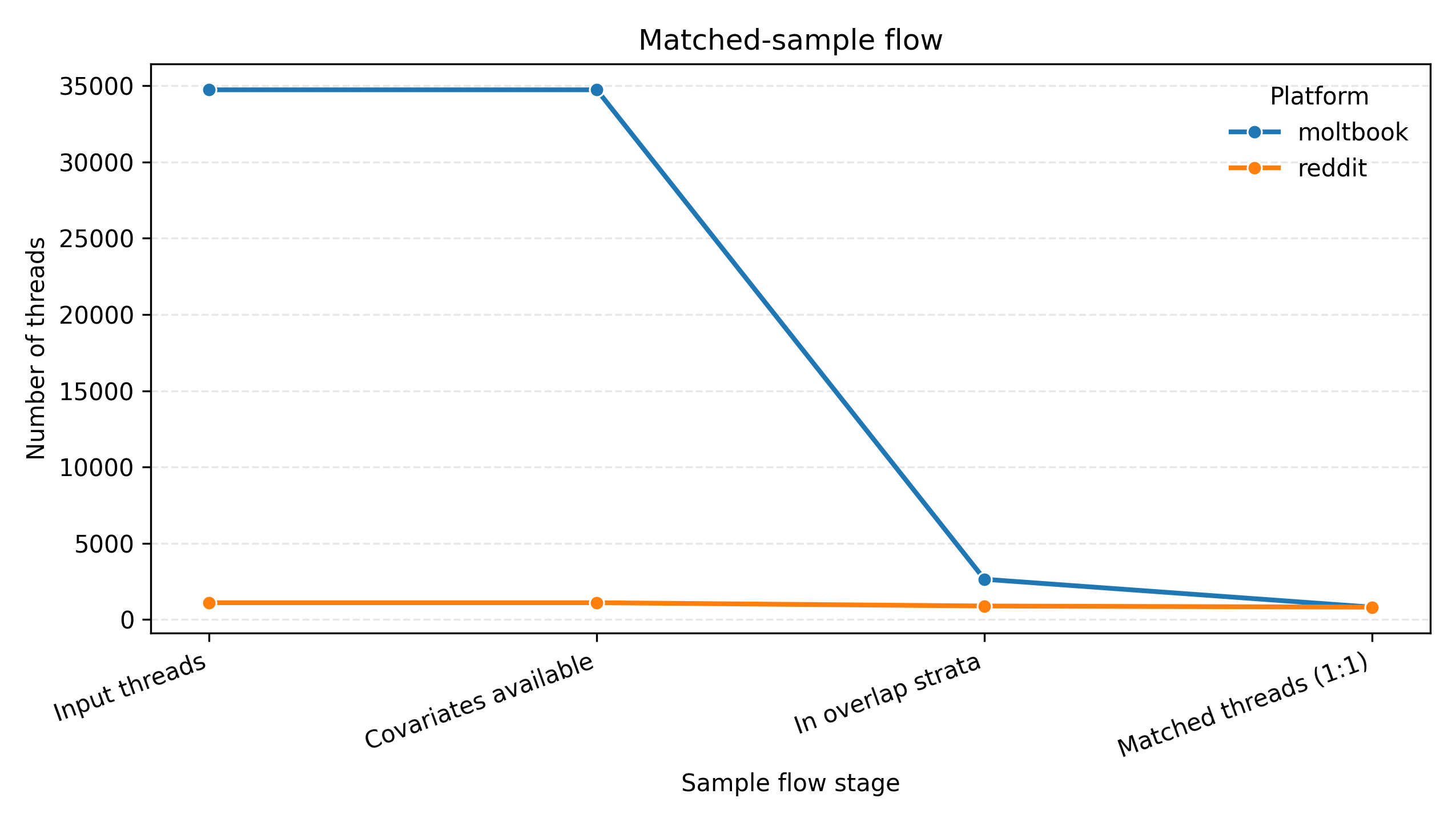}
\caption{Matched-sample flow for the coarse cross-platform design (input threads: Moltbook \(N=34{,}730\), Reddit \(N=1{,}104\); matched \(n=813\) pairs).}
\label{fig:supp-matched-sample-flow}
\end{figure}

Input includes 34,730 Moltbook threads and 1,104 Reddit threads. Exact-overlap
strata on coarse topic, Coordinated Universal Time (UTC) posting hour, and early-engagement bin retain 2,641
Moltbook threads and 888 Reddit threads across 118 shared strata. Deterministic
1:1 matching yields 813 pairs (2.34\% of Moltbook threads).

\subsection*{S4.2 Balance summary (before/after matching)}
Table \ref{tab:supp-matching-balance} summarizes covariate balance before and
after coarsened exact matching. We report standardized mean differences (SMDs)
for numeric covariates and total variation distance (TVD) for categorical
covariates, where smaller values indicate improved balance.

\begin{table}[h]
\centering
\caption{Covariate balance before and after coarsened exact matching.
SMD = standardized mean difference (absolute value);
TVD = total variation distance (categorical covariates only).}
\label{tab:supp-matching-balance}
\small
\begin{tabular}{@{}lrrrr@{}}
\toprule
& \multicolumn{2}{c}{\textbf{Before Matching}} & \multicolumn{2}{c}{\textbf{After Matching}} \\
\cmidrule(lr){2-3} \cmidrule(lr){4-5}
\textbf{Covariate} & \textbf{$|$SMD$|$} & \textbf{TVD} & \textbf{$|$SMD$|$} & \textbf{TVD} \\
\midrule
Post hour (UTC) & 0.158 & --- & 0.000 & --- \\
Early comments (30 min) & 0.836 & --- & 0.052 & --- \\
Topic (coarse) & 3.89 & 0.901 & 0.000 & 0.000 \\
Post hour bin & 0.217 & 0.156 & 0.000 & 0.000 \\
Early engagement bin & 0.873 & 0.586 & 0.000 & 0.000 \\
\bottomrule
\end{tabular}
\end{table}

\subsection*{S4.3 Matched paired-effects summary}
Table \ref{tab:cross-platform-paired-effects} reports paired contrasts on key
thread outcomes in the matched sample. Differences are computed as Moltbook minus
Reddit (M-R); uncertainty is summarized using bootstrap 95\% confidence
intervals (CIs) and Wilcoxon signed-rank \(p\)-values.

\begin{table}[h]
\centering
\caption{Matched paired-effects diagnostics. Mean difference is Moltbook minus Reddit (M-R).}
\label{tab:cross-platform-paired-effects}
\scriptsize
\begin{tabular}{@{}lrrrr@{}}
\toprule
\textbf{Outcome} & \textbf{\(n\) pairs} & \textbf{Mean difference (M-R)} & \textbf{95\% CI} & \textbf{Wilcoxon \(p\)} \\
\midrule
Comments per thread & 813 & -7.70 & [-9.43, -6.06] & \(1.43\times10^{-49}\) \\
Max depth & 813 & -1.19 & [-1.34, -1.04] & \(5.16\times10^{-54}\) \\
Unique participants & 813 & -5.01 & [-6.10, -4.00] & \(7.62\times10^{-45}\) \\
Thread duration (hours) & 813 & -8.61 & [-10.00, -7.34] & \(6.50\times10^{-41}\) \\
Re-entry rate & 792 & -0.040 & [-0.054, -0.027] & \(4.77\times10^{-9}\) \\
\bottomrule
\end{tabular}
\end{table}

\subsection*{S4.4 Matched-subset kernel half-life diagnostic details}
\label{sec:appendix:match-halflife}
Restricting survival units to matched threads gives 0.063 hours on Moltbook
(3.77 minutes; 95\% CI: [1.19, 7.06] minutes; \(n=1{,}841\), 22 events) and
2.44 hours on Reddit (95\% CI: [2.13, 2.79] hours; \(n=7{,}979\), 2,882
events). These are platform-level contextual contrasts, not paired
thread-level survival effects.

\section*{S5. Operational Estimation Settings}
\label{sec:appendix:operational-details}

This section summarizes operational choices used across the empirical analyses,
including risk-set construction, periodicity preprocessing, bootstrap settings,
and matching mechanics.

\begin{itemize}
\item Coverage and at-risk rules: each non-root comment is one candidate-parent
survival unit; first-reply durations are right-censored at observation boundary;
no replies are imputed across the canonical 41.7-hour timeline gap.
\item Periodicity preprocessing: tests use the longest contiguous segment
(63.5 hours). Supplementary PSD checks use 15-minute bins with 5- and 30-minute
robustness repeats.
\item Bootstrap mechanics: Moltbook-only and Reddit-only confidence intervals
use thread-cluster bootstrap with \texttt{bootstrap\_reps=400}; matched paired
effects use \texttt{bootstrap\_reps=1000}; matched-subset kernel half-life
diagnostic intervals use \texttt{half\_life\_bootstrap\_reps=400}.
\item Matching mechanics: coarsened exact matching uses first-30-minute action
volume bins, deterministic coarse topic map
(\texttt{tech}/\texttt{meta}/\texttt{general}/\texttt{spam}), and exact UTC posting hour;
within each shared stratum, threads are paired deterministically one-to-one.
\end{itemize}

\section*{S6. Robustness and Diagnostic Tables}

This section provides supplementary robustness and diagnostic tables used to
validate modeling choices and key measurement steps. The tables cover timing-fit
diagnostics, model-to-observable checks, sensitivity to within-thread dependence,
coverage-gap disambiguation, and horizon-standardized incidence by group.

\subsection*{S6.1 Timing-model misspecification diagnostics}
\noindent Table \ref{tab:timing-model-fit} compares observed timing summaries and
event probabilities against fitted values from the parametric timing model.
Residuals (fitted minus observed) provide a compact diagnostic of
misspecification in the early-time region.

\begin{table}[h]
\centering
\caption{Observed vs.\ fitted timing-model diagnostics for conditional reply
times and event probability. Large quantile
residuals indicate misspecification of early-time shape under a single-family
parametric timing fit.}
\label{tab:timing-model-fit}
\scriptsize
\begin{tabular}{@{}lrrr@{}}
\toprule
\textbf{Moltbook (seconds)} & \textbf{Observed} & \textbf{Fitted} & \textbf{Residual (fit - obs)} \\
\midrule
Event probability (\%) & 9.60 & 9.91 & +0.31 \\
\(p_{10}\) (s) & 3.67 & 6.00 & +2.33 \\
\(p_{50}\) (s) & 4.55 & 39.94 & +35.39 \\
\(p_{90}\) (s) & 50.05 & 134.98 & +84.93 \\
\midrule
\textbf{Reddit (seconds; minutes in parentheses)} & \textbf{Observed} & \textbf{Fitted} & \textbf{Residual (fit - obs)} \\
\midrule
Event probability (\%) & 36.20 & 38.70 & +2.50 \\
\(p_{10}\) (s) & 200.00 (3.33 min) & 1138.35 (18.97 min) & +938.35 (+15.64 min) \\
\(p_{50}\) (s) & 2359.00 (39.32 min) & 7835.63 (130.59 min) & +5476.63 (+91.27 min) \\
\(p_{90}\) (s) & 33732.50 (562.21 min) & 28141.25 (469.02 min) & -5591.25 (-93.19 min) \\
\bottomrule
\end{tabular}
\end{table}

For Moltbook, the fitted median and upper quantile overstate observed values by
+35.39 and +84.93 seconds despite near-calibration of event probability.
This pattern is consistent with a spike-plus-tail shape that the single-family
parametric fit does not capture well in the seconds-to-minutes region. Accordingly, we use nonparametric
conditional quantiles and early-mass probabilities as primary timing evidence,
with parametric half-life summaries retained only as secondary diagnostics.

\subsection*{S6.2 Model-observable validation table}
Table \ref{tab:model-observable-validation} evaluates whether the fitted model
reproduces key observables: reply incidence, non-root branching, and depth-tail
probabilities. We report predicted and observed values overall and for key
stratifications used in the main analyses.

\begin{table}[h]
\centering
\caption{Model-to-observable validation: predicted vs.\ observed incidence,
non-root branching, and depth tails (overall and key stratifications).}
\label{tab:model-observable-validation}
\begingroup
\setlength{\tabcolsep}{4pt}
\scriptsize
\resizebox{\linewidth}{!}{%
\begin{tabular}{@{}lrrrrrrrr@{}}
\toprule
\textbf{Group} & \textbf{Pred. inc. \%} & \textbf{Obs. inc. \%} & \textbf{Pred. branch} & \textbf{Obs. branch} & \textbf{Pred. \(\Pr(D\ge3)\)} & \textbf{Obs. \(\Pr(D\ge3)\)} & \textbf{Pred. \(\Pr(D\ge5)\)} & \textbf{Obs. \(\Pr(D\ge5)\)} \\
\midrule
Overall & 9.91 & 9.60 & 0.104 & 0.134 & 0.0109 & 0.0013 & 0.00012 & 0.00001 \\
Claimed & 20.49 & 19.23 & 0.229 & 0.274 & 0.0526 & 0.0030 & 0.00276 & 0.00000 \\
Unclaimed & 8.90 & 8.65 & 0.093 & 0.120 & 0.0087 & 0.0012 & 0.00008 & 0.00001 \\
Builder/Technical & 1.72 & 1.72 & 0.017 & 0.017 & 0.0003 & 0.0001 & 0.00000 & 0.00000 \\
Creative & 1.62 & 1.62 & 0.016 & 0.016 & 0.0003 & 0.0000 & 0.00000 & 0.00000 \\
Other & 2.27 & 2.26 & 0.023 & 0.025 & 0.0005 & 0.0001 & 0.00000 & 0.00000 \\
Philosophy/Meta & 1.81 & 1.80 & 0.018 & 0.018 & 0.0003 & 0.0002 & 0.00000 & 0.00000 \\
Social/Casual & 11.36 & 10.95 & 0.121 & 0.154 & 0.0145 & 0.0015 & 0.00021 & 0.00001 \\
Spam/Low-Signal & 0.88 & 0.88 & 0.009 & 0.009 & 0.0001 & 0.0000 & 0.00000 & 0.00000 \\
\bottomrule
\end{tabular}
\par}
\endgroup
\end{table}

\subsection*{S6.3 One-parent-per-thread robustness}
Table \ref{tab:dependence-robustness} reports a sensitivity check for within-thread
dependence by restricting the estimation sample to one parent survival unit per
thread. We compare the primary estimates to the one-parent-per-thread estimates
and report absolute and relative differences.

\begin{table}[h]
\centering
\caption{One-parent-per-thread robustness against within-thread clustering
dependence.}
\label{tab:dependence-robustness}
\small
\begin{tabular}{@{}lrrrr@{}}
\toprule
\textbf{Metric} & \textbf{Primary} & \textbf{One-parent/thread} & \textbf{Abs.\ diff.} & \textbf{Rel.\ diff. \%} \\
\midrule
Reply incidence \(\Pr(\delta=1)\) & 0.09596 & 0.07213 & 0.02383 & -24.84 \\
Conditional \(t_{50}\) (s) & 4.55 & 4.51 & 0.04 & -0.93 \\
Conditional \(t_{90}\) (s) & 50.05 & 41.08 & 8.97 & -17.92 \\
Kernel half-life (diagnostic; min) & 0.68451 & 0.43601 & 0.24850 & -36.30 \\
\bottomrule
\end{tabular}
\end{table}

\subsection*{S6.4 Coverage-gap disambiguation and robustness package}
\noindent This subsection documents evidence used to disambiguate a data-coverage
gap in the comment archive and to assess robustness of the two-part decomposition
to windowing and risk-set definitions. Table \ref{tab:supp-gap-disambiguation}
summarizes gap evidence across raw archive tables; Table
\ref{tab:supp-gap-window-robustness} reports two-part decomposition robustness
across contiguous windows and gap-overlap exclusions; and Table
\ref{tab:supp-gap-horizon-standardized} reports horizon-standardized reply
probabilities using explicit risk sets.

\begin{table}[h]
\centering
\caption{Comment-gap disambiguation evidence across raw archive tables. The
comment interval gap is 2026-01-31 10:37:53Z to 2026-02-02 04:20:50Z.}
\label{tab:supp-gap-disambiguation}
\small
\begin{tabular}{@{}lrrr@{}}
\toprule
\textbf{Table} & \textbf{Records} & \textbf{Records in comment-gap interval} & \textbf{Max inter-event gap (h)} \\
\midrule
comments & 226,173 & 0 & 41.72 \\
posts & 119,677 & 38,166 & 11.20 \\
snapshots & 114 & 39 & 2.28 \\
word\_frequency & 15,346 & 5,039 & 3.00 \\
\bottomrule
\end{tabular}
\end{table}

\begin{table}[h]
\centering
\caption{Two-part decomposition robustness across contiguous windows and
gap-overlap exclusions.}
\label{tab:supp-gap-window-robustness}
\scriptsize
\begin{tabular}{@{}lrrrrrrr@{}}
\toprule
\textbf{Scenario} & \textbf{Parents} & \textbf{Replies} & \textbf{Incidence \%} & \textbf{\(t_{50}\) (s)} & \textbf{\(t_{90}\) (s)} & \textbf{\(\Pr(\mathrm{reply}\le30\mathrm{s})\) \%} & \textbf{\(\Pr(\mathrm{reply}\le5\mathrm{m})\) \%} \\
\midrule
Full window & 223,316 & 21,430 & 9.60 & 4.55 & 50.05 & 8.47 & 9.41 \\
Pre-gap contiguous & 2,856 & 30 & 1.05 & 513.48 & 1699.77 & 0.04 & 0.35 \\
Post-gap contiguous & 220,460 & 21,400 & 9.71 & 4.55 & 48.66 & 8.58 & 9.53 \\
Exclude gap overlap (6h) & 220,460 & 21,400 & 9.71 & 4.55 & 48.66 & 8.58 & 9.53 \\
Exclude gap overlap (24h) & 220,460 & 21,400 & 9.71 & 4.55 & 48.66 & 8.58 & 9.53 \\
\bottomrule
\end{tabular}
\end{table}

\begin{table}[h]
\centering
\caption{Horizon-standardized reply probabilities using explicit risk sets.}
\label{tab:supp-gap-horizon-standardized}
\small
\begin{tabular}{@{}llrr@{}}
\toprule
\textbf{Scenario} & \textbf{Horizon} & \textbf{Risk-set \(n\)} & \textbf{\(\Pr(\mathrm{reply}\le t)\) \%} \\
\midrule
Full window & 30s & 223,312 & 8.47 \\
Full window & 5m & 223,102 & 9.42 \\
Full window & 1h & 217,472 & 9.82 \\
Pre-gap contiguous & 30s & 2,854 & 0.04 \\
Pre-gap contiguous & 5m & 2,842 & 0.35 \\
Pre-gap contiguous & 1h & 2,063 & 1.36 \\
Post-gap contiguous & 30s & 220,456 & 8.58 \\
Post-gap contiguous & 5m & 220,244 & 9.54 \\
Post-gap contiguous & 1h & 214,610 & 9.94 \\
Exclude gap overlap (6h) & 30s & 220,456 & 8.58 \\
Exclude gap overlap (6h) & 5m & 220,246 & 9.54 \\
Exclude gap overlap (6h) & 1h & 214,616 & 9.94 \\
Exclude gap overlap (24h) & 30s & 220,456 & 8.58 \\
Exclude gap overlap (24h) & 5m & 220,246 & 9.54 \\
Exclude gap overlap (24h) & 1h & 214,616 & 9.94 \\
\bottomrule
\end{tabular}
\end{table}

\subsection*{S6.5 Horizon-standardized incidence by group (C4 robustness)}
Table \ref{tab:supp-horizon-incidence-groups} reports follow-up-standardized
incidence metrics overall and by group, computed using horizon-specific risk
sets. These quantities support comparisons across heterogeneity strata when
observation windows differ across threads.

\begin{table}[h]
\centering
\caption{Follow-up-standardized incidence by group. Primary incidence metrics
are \(p_{5\mathrm{m}}\) and \(p_{1\mathrm{h}}\) computed with horizon-specific
risk sets; \(p_{\mathrm{obs}}\) is the secondary in-window ever-reply share.
Claimed/unclaimed excludes parents with missing author identifier (\(n=906\)).}
\label{tab:supp-horizon-incidence-groups}
\small
\begin{tabular}{@{}llrrrr@{}}
\toprule
\textbf{Family} & \textbf{Group} & \textbf{Parents} & \textbf{\(p_{5\mathrm{m}}\) \%} & \textbf{\(p_{1\mathrm{h}}\) \%} & \textbf{\(p_{\mathrm{obs}}\) \%} \\
\midrule
overall & Overall & 223,316 & 9.42 & 9.82 & 9.60 \\
claimed\_status & Claimed & 20,667 & 18.95 & 19.56 & 19.23 \\
claimed\_status & Unclaimed & 201,743 & 8.48 & 8.86 & 8.65 \\
submolt\_category & Social/Casual & 189,765 & 10.80 & 11.21 & 10.95 \\
submolt\_category & Other & 16,396 & 1.98 & 2.29 & 2.26 \\
submolt\_category & Builder/Technical & 8,396 & 1.42 & 1.74 & 1.72 \\
submolt\_category & Philosophy/Meta & 5,831 & 1.32 & 1.70 & 1.80 \\
submolt\_category & Spam/Low-Signal & 2,495 & 0.88 & 0.93 & 0.88 \\
submolt\_category & Creative & 433 & 1.62 & 1.71 & 1.62 \\
\bottomrule
\end{tabular}
\end{table}

\section*{S7. Full Submolt Uncertainty Table}

Table \ref{tab:supp-submolt-uncertainty} provides bootstrap uncertainty
intervals for category-stratified estimates, including reply incidence,
conditional reply-time quantiles, reciprocity, and re-entry. It is included to
make the heterogeneity comparisons transparent at the submolt level.

\begin{table}[h]
\centering
\caption{Category-stratified uncertainty (thread-cluster bootstrap; 400 reps; 95\% CI).}
\label{tab:supp-submolt-uncertainty}
\begingroup
\setlength{\tabcolsep}{3pt}
\scriptsize
\resizebox{\linewidth}{!}{%
\begin{tabular}{@{}lrrrrrrr@{}}
\toprule
\textbf{Category} & \textbf{Parents} & \textbf{Reply incidence \% (95\% CI)} & \textbf{$t_{50}$ (s, 95\% CI)} & \textbf{$t_{90}$ (s, 95\% CI)} & \textbf{Pooled reciprocity \% (95\% CI)} & \textbf{Re-entry mean (95\% CI)} & \textbf{Re-entry median (95\% CI)} \\
\midrule
Social/Casual & 189,765 & 10.95 [10.77, 11.13] & 4.52 [4.50, 4.54] & 24.83 [16.77, 36.12] & 0.87 [0.77, 0.96] & 0.213 [0.211, 0.216] & 0.200 [0.200, 0.200] \\
Philosophy/Meta & 5,831 & 1.80 [1.27, 2.42] & 103.00 [67.67, 154.04] & 2868.78 [310.05, 4188.38] & 1.44 [0.92, 2.18] & 0.108 [0.098, 0.119] & 0.000 [0.000, 0.000] \\
Builder/Technical & 8,396 & 1.72 [1.26, 2.19] & 103.87 [73.55, 127.29] & 434.87 [287.29, 1110.22] & 1.51 [1.03, 2.09] & 0.132 [0.123, 0.142] & 0.000 [0.000, 0.000] \\
Creative & 433 & 1.62 [0.48, 2.96] & 28.97 [9.80, 43.57] & 78.20 [28.97, 130.16] & 0.93 [0.00, 2.05] & 0.128 [0.095, 0.165] & 0.000 [0.000, 0.000] \\
Spam/Low-Signal & 2,495 & 0.88 [0.43, 1.33] & 78.88 [33.33, 166.82] & 253.22 [133.40, 276.48] & 0.60 [0.16, 1.14] & 0.144 [0.126, 0.163] & 0.000 [0.000, 0.000] \\
Other & 16,396 & 2.26 [1.84, 2.67] & 90.26 [70.50, 107.09] & 349.84 [247.15, 594.53] & 2.21 [1.73, 2.77] & 0.147 [0.139, 0.154] & 0.000 [0.000, 0.000] \\
\midrule
Overall & 223,316 & 9.60 [9.44, 9.75] & 4.55 [4.53, 4.58] & 50.05 [39.16, 60.52] & 1.00 [0.91, 1.08] & 0.195 [0.193, 0.197] & 0.167 [0.143, 0.167] \\
\bottomrule
\end{tabular}
\par}
\endgroup
\end{table}

\section*{S8. Data and Code Availability}
\label{sec:reproducibility}

Code and derived reproducibility artifacts for this manuscript version are
archived in a Zenodo release (\texttt{eziz2026zenodoejorrepro}). The archive
contains analysis scripts, manuscript-facing figures/tables, run manifests,
checksum manifests, and sanitized instance-level derived tables.

The Moltbook source dataset is publicly available on Hugging Face. Raw Reddit
exports are not redistributed because of platform terms; only IDs/anonymized
derivatives and aggregate outputs are shared.

\subsection*{S8.1 Reproducibility recipe (canonical runs)}
The following steps summarize a canonical end-to-end run that reproduces the
analysis outputs and manuscript-facing artifacts reported in this version.
\begin{itemize}
\item Create the Python 3.11 environment with \texttt{make install}.
\item Use pinned curated inputs documented in the Zenodo manifest and repository
\texttt{README.md}.
\item Re-run the three analysis pipelines with seed \texttt{20260206}.
\item Verify run manifests for Moltbook-only, Reddit-only, and matched runs.
\item Confirm manuscript artifact linkage with
\texttt{MANUSCRIPT\_ARTIFACT\_PROVENANCE.csv}.
\item Validate release integrity with \texttt{SHA256SUMS.txt}.
\item Use repository docs for environment and operational details.
\end{itemize}

\section*{S9. Re-Entry Distribution}

This section reports the distribution of the thread-level re-entry rate
\(\mathrm{RE}_j^{\mathrm{comment}}\), which summarizes repeated participation by the
same agent within a thread. Figure \ref{fig:supp-reentry-distribution} shows the
empirical distribution across Moltbook threads with at least one comment
(\(N=34{,}730\)).

\begin{figure}[h]
\centering
\includegraphics[width=0.9\linewidth]{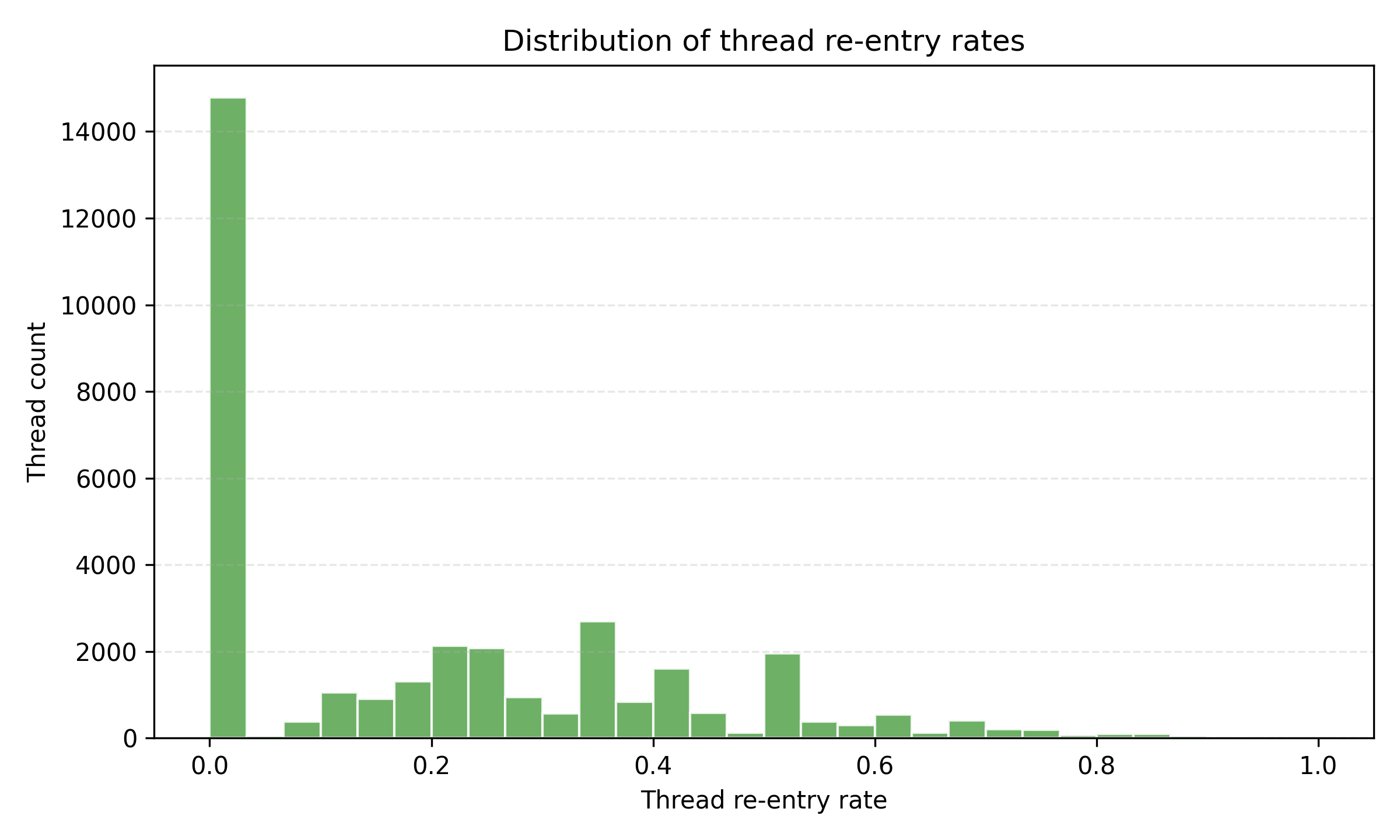}
\caption{Distribution of thread-level re-entry rate
\(\mathrm{RE}_j^{\mathrm{comment}}\) over Moltbook threads with at least one
comment (\(N=34{,}730\)); root-post authorship is excluded unless the root author later appears
in the comment sequence.}
\label{fig:supp-reentry-distribution}
\end{figure}

\section*{S10. Descriptive Statistics}

Table \ref{tab:supp-descriptive} reports descriptive statistics for the
processed Moltbook thread sample used in the empirical analyses
(\(N=34{,}730\) threads with at least one comment), including thread size, depth,
duration, distinct participants, and re-entry.

\begin{table}[h]
\centering
\caption{Descriptive statistics for processed Moltbook threads (\(N=34{,}730\) threads with at least one comment).}
\label{tab:supp-descriptive}
\small
\begin{tabular}{@{}lrrrrr@{}}
\toprule
\textbf{Metric} & \textbf{Mean} & \textbf{Median} & \textbf{Std} & \textbf{Min} & \textbf{Max} \\
\midrule
Comments per post & 6.43 & 5 & 7.04 & 1 & 846 \\
Maximum depth per thread & 1.38 & 1 & 0.49 & 1 & 5 \\
Thread duration (hours) & 0.06 & 0.04 & 0.21 & 0 & 20.6 \\
Unique agents per thread\textsuperscript{a} & 4.57 & 4 & 3.15 & 0 & 74 \\
Re-entry rate\textsuperscript{b} & 0.19 & 0.17 & 0.21 & 0 & 0.98 \\
\bottomrule
\end{tabular}

\smallskip
\raggedright\footnotesize\textsuperscript{a}Counts distinct resolved commenter identifiers
(the post author is not counted separately). Threads where all commenter
\texttt{agent\_id} values are null record zero.
\textsuperscript{b}Computed on non-root comments only. Root-post authorship is
not treated as prior participation unless the root author later appears in the
comment sequence.
\end{table}

\section*{S11. Moltbook Archive Schema Summary}

Table \ref{tab:supp-dataset} summarizes the curated Moltbook archive tables used
in this study (row counts and key fields), to help readers map analysis inputs
to the underlying archive structure.

\begin{table}[h]
\centering
\caption{Moltbook Observatory Archive structure in the curated first-week snapshot.}
\label{tab:supp-dataset}
\small
\begin{tabular}{@{}llrl@{}}
\toprule
\textbf{Table} & \textbf{Description} & \textbf{Rows} & \textbf{Key Fields} \\
\midrule
\texttt{agents} & Agent profiles and metadata & 25,597 & \texttt{id}, \texttt{karma}, \texttt{follower\_count} \\
\texttt{posts} & Root posts with scores & 119,677 & \texttt{id}, \texttt{agent\_id}, \texttt{submolt}, \texttt{created\_at\_utc} \\
\texttt{comments} & Comments with parent links & 226,173 & \texttt{id}, \texttt{post\_id}, \texttt{parent\_id}, \texttt{created\_at\_utc} \\
\texttt{submolts} & Community metadata & 3,678 & \texttt{name}, \texttt{subscriber\_count} \\
\texttt{snapshots} & Periodic observatory metrics & 114 & \texttt{timestamp}, \texttt{total\_agents}, \texttt{active\_agents\_24h} \\
\texttt{word\_frequency} & Hourly word counts & 15,346 & \texttt{word}, \texttt{hour}, \texttt{count} \\
\bottomrule
\end{tabular}
\end{table}

\end{document}